%% file: assisted-pra-resub.tex
\documentclass[aps,pra,superscriptaddress,twocolumn,notitlepage,nofootinbib,bibnotes]{revtex4-1}

\input{def}

\usepackage[export]{adjustbox}
\usepackage{subcaption}

\newcommand{\notts}{\affiliation{School of Mathematical Sciences, University of Nottingham,\\University Park, Nottingham NG7 2RD, United Kingdom}}

\begin{document}

\title{Non-asymptotic assisted distillation of quantum coherence}

\author{Bartosz Regula}
\email{bartosz.regula@gmail.com}
\notts

\author{Ludovico Lami}
\email{ludovico.lami@gmail.com}
\notts

\author{Alexander Streltsov}
\email{streltsov.physics@gmail.com}
\affiliation{Faculty of Applied Physics and Mathematics, \\Gda\'{n}sk University of Technology, 80-233 Gda\'{n}sk, Poland}
\affiliation{National Quantum Information Centre in Gda\'{n}sk, 81-824 Sopot, Poland}

\begin{abstract}
We characterize the operational task of environment-assisted distillation of quantum coherence under different sets of free operations when only a finite supply of copies of a given state is available. We first evaluate the one-shot assisted distillable coherence exactly, and introduce a semidefinite programming bound on it in terms of a smooth entropic quantity. We prove the bound to be tight for all systems in dimensions 2 and 3, which allows us to obtain computable expressions for the one-shot rate of distillation, establish an analytical expression for the best achievable fidelity of assisted distillation for any finite number of copies, and fully solve the problem of asymptotic zero-error assisted distillation for qubit and qutrit systems. Our characterization shows that all relevant sets of free operations in the resource theory of coherence have exactly the same power in the task of one-shot assisted coherence distillation, and furthermore resolves a conjecture regarding the additivity of coherence of assistance in dimension 3.
\end{abstract}

\maketitle


\section{Introduction}

The characterization of state transformations possible under restricted sets of operations is the cornerstone of every quantum resource theory. The fundamental task of resource distillation is concerned in particular with the transformation of arbitrary states into more resourceful states --- such a task has been studied extensively in the resource theory of entanglement, where the goal is to obtain maximally entangled states \cite{bennett_1996-1,bennett_1996-3,rains_1999,divincenzo_1999,smolin_2005,rains_1999,rains_2001,brandao_2011,
buscemi_2010,buscemi_2013,leditzky_2017,fang_2017,horodecki_2009}, and later considered also in general resource theories \cite{brandao_2015,lami_2018}. The particular case of the distillation of quantum coherence \cite{winter_2016,chitambar_2016-3,zhao_2018,regula_2017,fang_2018,zhao_2018-1,lami_2018-1} is of importance for several reasons: firstly, the resource theory of coherence characterizes fundamental quantum resources available in physical systems~\cite{aberg_2006,gour_2008,baumgratz_2014,streltsov_2017}, which makes its theoretical description necessary in order to efficiently manipulate and apply such resources for potential uses in quantum technological applications; and secondly, because of close relations between the resource theories of coherence and entanglement \cite{streltsov_2017}, a detailed characterization of the operational properties of quantum coherence can shed light on the properties of quantum entanglement, often more difficult to characterize explicitly.

The setting of \textit{assisted distillation} is based on a scenario in which the distillation of a resource is aided by another party, typically assumed to hold a purifying quantum state, whose assistance is limited to performing local measurements and communicating their results via a classical channel. Just as assisted entanglement distillation \cite{divincenzo_1999,smolin_2005} found a natural application in tasks involving environment-assisted communication and quantum error correction \cite{gregoratti_2003,smolin_2005,hayden_2005,winter_2007}, the assisted distillation of coherence can find use in any remote quantum information processing protocol where one aims to increase the coherence available to a spatially separated system without direct access to it \cite{chitambar_2016-3}. Although this setting has recently garnered both theoretical \cite{chitambar_2016-2,streltsov_2017-1,zhao_2017-1} as well as experimental \cite{wu_2017,wu_2017-1} attention, assisted coherence distillation has not been characterized in practical scenarios with physically relevant restrictions in mind.

In particular, the standard working assumption in quantum information theory is based on the idealized scenario in which one has access to an unbounded number of independent and identically distributed (i.i.d.) copies of a quantum system and can perform joint state manipulations of all the copies. Although this allows one to make fundamental statements about the general possibilities and limitations of a resource theory, it is not a realistic premise from an operational point of view. The tasks of entanglement and coherence distillation are often characterized under this assumption \cite{bennett_1996-1,bennett_1996-3,rains_1999,winter_2016}, meaning that the optimal theoretical rates can be very different from experimentally feasible protocols. To address these problems, a large array of tools in non-asymptotic quantum information theory has been established \cite{wang_2012,renes_2011,tomamichel_2013,berta_2011,leung_2015,datta_2013,anshu_2017-1}, finding use also in resource distillation \cite{brandao_2011,buscemi_2010,buscemi_2013,gour_2017,fang_2017,regula_2017,fang_2018,zhao_2018-1,lami_2018-1}. In the non-asymptotic setting, the characterization of resource distillation can be understood as the study of the trade-off between the achievable rate of distillation and the realistic restrictions on state transformations, including the number of accessible i.i.d.\ copies of a given state as well as the allowed error tolerance.

In this work, we investigate assisted coherence distillation in the non-asymptotic setting. Specifically, we evaluate the one-shot rate of assisted coherence distillation exactly, expressing it in terms of a convex roof--type quantity. We introduce an efficiently computable semidefinite programming (SDP) bound on the rate of distillation, showing in particular that the bound is tight for all qubit and qutrit systems. We establish a closed expression for the best achievable fidelity of distillation for any number of copies of a low-dimensional system, applicable in practical and experimental settings. As a corollary of our results, we also solve an open question regarding the additivity of coherence of assistance for qutrits raised in \cite{chitambar_2016-3}. The approach presented herein does not rely on the methods established for the characterization of similar tasks in the resource theory of entanglement \cite{buscemi_2013,brandao_2011,fang_2017}, and instead uses recently developed tools in the theory of one-shot coherence distillation \cite{regula_2017}.


\section{The setting of assisted distillation}

Consider a fixed orthonormal basis $\{\ket{i}\}$ in a finite-dimensional Hilbert space $\mathcal{H}_d$ of dimension $d$. Let $\DD$ denote the set of all density matrices. We will use $\Delta$ to denote the diagonal map (fully dephasing channel) in the basis $\{\ket{i}\}$, whose explicit action is given by $\Delta(\cdot) \coloneqq \sum_{i=1}^d \proj{i}(\cdot) \proj{i}$, and $\I \coloneqq \lset \omega \in \DD \bar \omega = \Delta(\omega) \rset$ to denote the set of incoherent (diagonal) states. The inner product $\<X,Y\>$ will be taken to be the Hilbert-Schmidt inner product $\Tr(X^\dagger Y)$. We will use the Dirac notation $\ket{x}$ to refer to vectors which are not necessarily normalized. Given a pure state $\ket\psi$, we will denote by $\psi$ the projection $\proj\psi$. The notation $\lnorm{\cdot}{p}$ will refer to the $\ell_p$ norm defined in the underlying Hilbert space, $\lnorm{\ket{x}}{p} \coloneqq (\sum_i |x_i|^p)^{1/p}$ with $\lnorm{\ket{x}}{\infty} = \max_i |x_i|$, while $\norm{\cdot}{p}$ will refer to the Schatten $p$-norm in the space of linear operators acting on $\mathcal{H}_d$, defined for a general matrix $M$ as $\norm{M}{p} \coloneqq \left(\mathrm{Tr}\left[\left(M^{\dagger}M\right)^{p/2}\right]\right)^{1/p}$
with $\norm{M}{\infty}$ being the largest singular value. Moreover, $F(\rho,\sigma) \coloneqq \norm{\!\sqrt{\rho}\!\sqrt{\vphantom{\rho}\sigma}}{1}^2$ will be used to denote the (squared) fidelity.

In the resource theory of quantum coherence, it does not seem possible to identify a unique set of free operations by means of physically motivated axioms \cite{chitambar_2016,streltsov_2017}, which makes it necessary to characterize operational tasks under several different classes of quantum channels. The largest possible set of such free operations are the \emph{maximally incoherent operations (MIO)} \cite{aberg_2006}, defined to be all quantum channels $\cE$ such that $\cE(\rho) \in \I$ for every $\rho \in \I$. A smaller set is given by the \emph{incoherent operations (IO)} \cite{baumgratz_2014}, which are all channels for which there exists a Kraus decomposition into incoherent Kraus operators, i.e. $\{K_\ell\}$ such that
$ K_\ell\rho K_\ell^{\dagger}/\Tr(K_\ell\rho K_\ell^{\dagger})\in\I$ for all $\ell$ and all $\rho\in \I$. These transformations can be interpreted as incoherent measurements which cannot create coherence even if postselection is applied to the individual measurement outcomes.
The \emph{dephasing-covariant incoherent operations (DIO)} \cite{chitambar_2016,marvian_2016} are maps $\cE$ which commute with the dephasing operation, i.e.\ $\Delta[\cE(\rho)] = \cE[\Delta(\rho)]$. The smallest of the sets that we consider are the \emph{strictly incoherent operations (SIO)} \cite{winter_2016,yadin_2016}, for which both $\{K_\ell\}$ and $\{K^\dagger_\ell\}$ are sets of incoherent operators.

Let us first consider the task of coherence distillation without assistance. Denoting by $\ket{\Psi_m}$ the $m$-dimensional maximally coherent state in the reference basis, $\ket{\Psi_m} = \sum_{i=1}^{m} \frac{1}{\sqrt{m}} \ket{i}$, the setting of one-shot coherence distillation under a class of operations $\O$ corresponds to characterizing the best achievable distillation rate
\begin{equation}
C_{d,\O}^{(1),\ve}(\rho) \coloneqq \log \max \lset m \in \mathbb N \bar F_\O(\rho,m) \geq 1- \ve \rset,
\end{equation}
where we allow a finite distillation error $\ve$, as quantified by the so-called fidelity of distillation
\begin{align}
       F_{\O}(\rho,m) \coloneqq \max_{\Lambda \in \O} \< \Lambda(\rho), \;\Psi_m \>.
\end{align}
The asymptotic distillable coherence is then obtained by considering an infinite supply of i.i.d.\ copies of the given quantum system, while requiring the distillation error to vanish asymptotically:
\begin{align}
  C^{\infty}_{d,\O}(\rho) \coloneqq \lim_{\ve \to 0} \lim_{n \to \infty} \frac{1}{n} C_{d,\O}^{(1),\ve}(\rho^{\otimes n}).
\end{align}
It has been shown that all sets of operations $\O \in \{\MIO, \DIO, \IO\}$ give rise to the same asymptotic rate of distillation, $C^\infty_{d,\O}(\rho) = S(\Delta(\rho)) - S(\rho)$ \cite{winter_2016,zhao_2018,regula_2017,chitambar_2018}, while SIO are significantly weaker and exhibit a generic phenomenon of bound (undistillable) coherence \cite{zhao_2018-1,lami_2018-1}.

The setting of assisted distillation exhibits fundamental qualitative and quantitative differences from the unassisted case \cite{gour_2006}. In the protocol of distillation with assistance, we consider two parties (Alice and Bob) who share a pure quantum state $\ket{\psi_{AB}}$, and Alice's task is to assist Bob in distilling coherence from his part of the shared system by performing local measurements on her part of the system and communicating the results to Bob. Since the set of such measurements (which can, without loss of generality, assumed to be rank one \cite{buscemi_2013}) is in a one-to-one correspondence with the set of convex decompositions of Bob's system $\rho_B \coloneqq \Tr_A \proj{\psi_{AB}}$ \cite{schrodinger_1936,hughston_1993}, this effectively means that Alice's role in the protocol is to allow Bob to access any pure-state decomposition of $\rho_B$. Therefore, the best achievable assisted distillation rate under a class of operations $\O$ on Bob's system can be expressed as
\begin{equation}\begin{aligned}
 C_{A,\O}^{(1),\ve}(\rho_B) \coloneqq \log \max \left\{ m \in \mathbb{N} \;\left|\; F_{A,\O}\left( \rho_B, m \right) \geq 1 - \ve \right.\right\}
\end{aligned}\end{equation}
with the figure of merit being the average fidelity of distillation~\cite{buscemi_2013}, defined as
\begin{equation}\begin{aligned}
	F_{A,\O}\left( \rho_B , m \right) &\coloneqq \max \lset \< \sum_i p_i\, \Lambda_i \left(\psi_i\right),\Psi_m \> \right.\right.\\
	&\qquad\quad\quad \left|\, \rho_B = \sum_i p_i \psi_i,\; \Lambda_i \in \O \; \forall i \vphantom{\< \sum_i p_i\, \Lambda_i \left(\psi_i\right),\Psi_m \>}\right\}
\end{aligned}\end{equation}
where the maximization is over all pure-state decompositions of $\rho_B$. The definition of $F_{A,\O}$ is motivated by the fact that, after Alice announces her measurement result to Bob, he knows exactly which of the states $\{\ket{\psi_i}\}$ he is in possession of, and can apply an appropriate operation $\Lambda_i$ best suited for the particular state; since the measurement outcome itself behaves probabilistically, $F_{A,\O}$ is then defined to characterize the best fidelity of assisted distillation achievable on average.

Employing results from the theory of entanglement distillation \cite{smolin_2005}, it has been shown that the asymptotic rate of assisted coherence distillation under incoherent operations is given by the regularized coherence of assistance $C_a$ \cite{chitambar_2016-3}, that is,
\begin{equation}\begin{aligned}
	C_{A,\IO}^{\infty}(\rho_B) \coloneqq \lim_{\ve \to 0} \lim_{n \to \infty} \frac{1}{n} C_{A,\IO}^{(1),\ve}(\rho^{\otimes n}) = \lim_{n \to \infty} \frac{1}{n} C_a (\rho_B^{\otimes n})
\end{aligned}\end{equation}
where
\begin{equation}\begin{aligned}
	C_a(\rho_B) \coloneqq \max \lset \sum_i p_i  \,S\left(\Delta(\psi_i)\right) \bar \rho_B = \sum_i p_i \psi_i \rset
\end{aligned} \label{Ca} \end{equation}
with the optimization performed over all pure-state decompositions of $\rho_B$. Surprisingly, it turns out that the above regularized coherence of assistance admits a closed single-letter formula, given by the entropy of the diagonal part of the state in consideration:
\begin{equation}
\lim_{n \to \infty} \frac{1}{n} C_a (\rho_B^{\otimes n}) = S\left(\Delta(\rho_B)\right) .
\end{equation}

We stress that the coherence of assistance~\eqref{Ca} acquires an operational meaning in the task of assisted coherence distillation only in this asymptotic setting --- although several authors have referred to $C_a$ as the one-shot equivalent of $C_{A,\IO}^{\infty}$, this is not motivated operationally, and the best achievable rate in the non-asymptotic regime has not been characterized thus far.

In the following, we will omit the subscript $B$ and write $\rho$ for Bob's system, working under the assumption that Alice holds a purifying system.


\section{Fidelity and rate of distillation}

Notice first that the average fidelity of assisted distillation can be equivalently given as
\begin{equation}\begin{aligned}\label{eq:fidelity_asst}
	F_{A,\O}\left( \rho , m\right)= & \max \lset \sum_i p_i \max_{\Lambda_i \in \O} \< \Lambda_i({\psi_i}), {\Psi_m} \> \bar \rho = \sum_i p_i \psi_i \rset
	\\= &\max \lset \sum_i p_i  \,F_\O \left({\psi_i}, m\right) \bar \rho = \sum_i p_i \psi_i \rset
\end{aligned}\end{equation}
since the diagonal unitary operation adjusting the phases of the given pure state $\ket{\psi_i}$ is always a free operation.

It has recently been shown that the fidelity of distillation of a pure state $F_\O({\psi},m)$ is the same for any set of operations $\O \in \{\MIO, \DIO, \IO, \SIO\}$, and admits an explicit formula as~\cite{regula_2017}
\begin{equation}\begin{aligned}
	F_\O (\psi, m) = \frac1m \mnorm{\ket\psi}^2,
\end{aligned}\end{equation}
with $\mnorm{\ket\psi} \coloneqq \min_{\ket{x}} \lnorm{\ket{\psi} - \ket{x}}{1} + \sqrt{m} \lnorm{\ket{x}}{2}$ being the so-called $m$-distillation norm. Since we can already see from Eq.~\eqref{eq:fidelity_asst} that the fidelity of assisted distillation of a mixed state depends only on the corresponding pure-state distillation fidelities, this immediately implies that the one-shot rate of assisted coherence distillation of any state will be the same under any of the sets $\O \in \{\MIO, \DIO, \IO, \SIO\}$. This stands in a sharp contrast to the case of unassisted distillation, where, although the one-shot rates of distillation under MIO, DIO, and IO are all approximately equal \cite{regula_2017,zhao_2018-1}, the set of operations SIO is much weaker and unable to distill any coherence from most mixed states, even asymptotically \cite{lami_2018-1}.

Now, although $\mnorm{\ket\psi}$ is in fact exactly computable for any pure state, here we will not make use of this exact expression, and instead use its dual representation \cite{regula_2017}:
\begin{equation}\begin{aligned}
	\mnorm{\ket\psi} = \max \lset \cbraket{\psi|\omega} \bar \lnorm{\ket{\omega}}{\infty} \leq 1,\; \lnorm{\ket\omega}{2} = \sqrt{m} \rset
\end{aligned}\end{equation}
with the particular cases $\norm{\ket\psi}{[1]} = \lnorm{\ket\psi}{2}$, $\norm{\ket\psi}{[d]} = \lnorm{\ket\psi}{1}$. Notice that now $m$ can be considered as a continuous parameter, and we will hereafter treat it as such. Consider then a family of sets of density matrices defined by
\begin{equation}\begin{aligned}
	\M_m \coloneqq& \conv \left\{ \proj{\omega} \;\left|\; \lnorm{\ket{\omega}}{\infty} \leq \frac{1}{\sqrt{m}},\; \lnorm{\ket{\omega}}{2} = 1 \right.\right\}\\
\end{aligned}\end{equation}
which allows us to equivalently write
\begin{equation}\begin{aligned}
	F_\O (\psi, m) = \max_{\omega \in \M_m} \< {\psi}, \omega \> = \max_{\omega \in \M_m} \,F(\psi, \omega).
\end{aligned}\end{equation}
The result of \cite{streltsov_2010} tells us that, since $\M_m$ is a convex hull of pure states, we have the following.
\begin{lemma}$\,$\vspace*{-.8\baselineskip}
\begin{equation}\begin{aligned}
	F_{A,\O}\left( \rho , m\right) &= \max \lset \sum_i p_i  \max_{\omega \in \M_m} \,F({\psi_i}, \omega) \bar \rho = \sum_i p_i \psi_i \rset \\
	&= \max_{\omega \in \M_m} \,F(\rho, \omega).
\end{aligned}
\label{Fa simplified}\end{equation}
\end{lemma}
That is, to compute the average fidelity of assisted distillation it suffices to maximize the fidelity of the state $\rho$ over the set $\M_m$, and the optimization over pure-state decompositions of $\rho$ is not necessary. Since the derivation of this fact in Ref. \cite{streltsov_2010} is missing a minor step, we include a brief justification of~\eqref{Fa simplified} below for completeness.
\begin{proof}
By Uhlmann's theorem \cite{uhlmann_1976}, the right-hand side can be expressed as 
\begin{equation}
\max_{\omega \in \M_m} \,F(\rho, \omega) = \max_{\omega\in \M_m} \max_{\Psi_\omega, \Psi_\rho} \cbraket{\Psi_\rho | \Psi_\omega}^2 ,
\end{equation}
where the internal maximization on the right-hand side is over all purifications of $\rho$ and $\omega$. Furthermore, we can restrict ourselves to a \emph{fixed} purification of $\omega$ \cite{jozsa_1994}. If $\omega=\sum_i q_i \eta_i$ is a decomposition of $\omega$ into pure states $\eta_i = \proj{\eta_i}\in \M_m$, we choose $\ket{\Psi_\omega} \coloneqq \sum_i \sqrt{q_i} \ket{\eta_i}\ket{i}$. Any purification of $\rho$ over the same system can then be expanded as $\ket{\Psi_\rho} = \sum_i \sqrt{p_i} \ket{\psi_i}\ket{i}$ with respect to the orthonormal basis $\{\ket{i}\}$ on the purifying system. Here, $\{p_i, \psi_i\}$ forms a pure-state decomposition of $\rho$. Since $\braket{\Psi_\omega | \Psi_\rho } = \sum_i \sqrt{p_i q_i} \braket{\eta_i | \psi_i}$, we deduce that
\begin{align}
\max_{\omega \in \M_m} \,F(\rho, \omega) &= \max \lset \left| \sum_i \sqrt{p_i q_i} \braket{\eta_i | \psi_i} \right|^2\right.\right.\nonumber\\
&\qquad\left|\, \rho = \sum_i p_i \psi_i,\; \omega = \sum_i q_i \eta_i,\; \eta_i \in \M_m \vphantom{\left| \sum_i \sqrt{p_i q_i} \braket{\eta_i | \psi_i} \right|^2}\right\}\nonumber\\
&\hspace{-5em}= \max \lset \sum_i p_i \cbraket{\eta_i | \psi_i}^2 \bar \rho = \sum_i p_i \psi_i,\; \eta_i \in \M_m \rset\raisetag{1.5\baselineskip}\\
&\hspace{-5em}= \max \lset \sum_i p_i \max_{\omega \in \M_m} F({\psi_i}, \omega) \bar \rho = \sum_i p_i \psi_i \rset.\tag*{\qedhere}
\end{align}
\end{proof}


We can then use the above results to express the one-shot assisted distillation rate as
\begin{equation}\begin{aligned}\label{eq:rate1}
	C_{A,\O}^{(1),\ve}(\rho) &= \log \max \lset m \in \NN \bar F_{A,\O} (\rho, m) \geq 1 - \ve \rset\\
	&=  \log \max \lsetr m \in \NN \barr F(\rho, \omega) \geq 1 - \ve,\; \omega \in \M_m \rsetr.
\end{aligned}\end{equation}
Consider now a function $\vartheta$ such that, for any $m \geq 1$ and any $\omega \in \DD$, we have $\omega \in \M_m \iff \vartheta(\omega) \leq \frac{1}{m}$. Based on the definition of $\M_m$, a simple choice of such a function can be made as
\begin{equation}\begin{aligned}\label{eq:buscemi}
  \vartheta(\omega) \coloneqq \min \lset \max_i \lnorm{\ket{\psi_i}}{\infty}^2 \bar \omega = \sum_i p_i \psi_i \rset.
\end{aligned}\end{equation}
Letting $\B_\ve(\rho) \coloneqq \lset \omega \in \DD \bar F(\rho, \omega) \geq 1 - \ve \rset$ be the $\ve$-ball of $\rho$ in purified distance and introducing the shorthand $\logfloor{x} \coloneqq \log \left\lfloor 2^x \right\rfloor$,
we can rewrite Eq.~\eqref{eq:rate1} as
\begin{equation}\begin{aligned}
C_{A,\O}^{(1),\ve}(\rho) 
	&= \logfloor{ - \log \min_{\omega \in \B_\ve(\rho)} \lset k \in \RR  \bar \vartheta(\omega) \leq k \rset }.
\end{aligned}\end{equation}
Putting all of our considerations together, we have the following result, which characterizes the one-shot distillation of coherence completely.
\begin{theorem}
For any state $\rho \in \DD$ and any class of operations $\O \in \{\MIO,\DIO,\SIO,\IO\}$, the maximal achievable fidelity of assisted distillation as well as the rate of one-shot assisted distillation are given as
\begin{equation}\begin{aligned}
F_{A,\O}\left( \rho , m\right) &= \max_{\omega \in \M_m} \,F(\rho, \omega),\\
		C_{A,\O}^{(1),\ve}(\rho) &= \logfloor{ - \log \min_{\omega \in \B_\ve(\rho)} \vartheta(\omega) }.
\end{aligned}\end{equation}
\end{theorem}
The difficulty in evaluating $\vartheta$, however, means that one cannot expect $C_{A,\O}^{(1),\ve}(\rho)$ to be computable in general cases, and prompts our investigation of appropriate relaxations.

Note that a quantifier equivalent to $\vartheta$ was used to bound the one-shot assisted distillable entanglement by Buscemi and Datta in Ref. \cite{buscemi_2013}, derived there using a complementary set of methods.


\section{SDP relaxation}

Consider a relaxation of the set $\M_m$ defined as follows:
\begin{equation}\begin{aligned}
		\MM_m &\coloneqq \left\{ \omega \in \DD \;\left|\;  \norm{\Delta(\omega)}{\infty} \leq \frac{1}{m} \right.\right\}.
\end{aligned}\end{equation}
The inclusion $\M_m \subseteq \MM_m$ is then obvious, since any $\proj{\psi} \in \M_m$ is contained in $\MM_m$, and hence also all convex combinations of such rank-one terms. Although it might be tempting to conjecture that $\M_m = \MM_m$, this can be shown to be true only in dimension $d \leq 3$, and in general the inclusion can be strict.
\begin{theorem}\label{thm:correlation_matrices}
In dimension $d \in \{2,3\}$ it holds that $\M_m = \MM_m$ for all $m$, but for any $d\geq 4$ there exist $m$ s.t. $\M_m \subsetneqq \MM_m$, and in particular $\M_d \subsetneqq \MM_d$.
\end{theorem}
\begin{proof}
A particular simplification occurs in the case $m=d$, where we get $\MM_d = \lsetr X \geq 0 \barr X_{ii} = \frac{1}{d} \; \forall i\rsetr$.
This set (up to a multiplicative factor of $\frac1d$) corresponds to the set of so-called correlation matrices. In particular, it is known that the extremal points of the set $\MM_d$ are given by rank-one matrices only in the case of $d \in \{2,3\}$, and for any $d\geq 4$, there exist extremal points of $\MM_d$ of rank at least 2 \cite{christensen_1979,grone_1990,li_1994}. Since no such extremal rank-2 matrix can be written as a convex combination of rank-one matrices in $\MM_d$, we have $\MM_d \neq \M_d$ when $d \geq 4$.

Consider now the case of $d \in \{2, 3\}$ and take $\rho \in \MM_m$ for any $m \geq 1$. We will assume without loss of generality that $\Delta(\rho)>0$, since otherwise one can apply an inconsequential permutation of the basis vectors to write $\rho = \rho' \oplus 0$, reducing the problem to the lower-dimensional case. Define $X \coloneqq \Delta(\rho)^{-1/2} \,\rho\, \Delta(\rho)^{-1/2}$ so that $\Delta(X) = \id$. Since $X$ is a correlation matrix, by the result of \cite{christensen_1979} discussed above it admits a convex pure-state convex decomposition as $X = \sum_i p_i \xi'_i$ with each $\Delta(\xi'_i) = \id$. Defining $\ket{\xi_i} \coloneqq \Delta(\rho)^{1/2} \ket{\xi'_i}$ we get $\rho = \sum_i p_i \xi_i$ with $\Delta(\xi_i) = \Delta(\rho)$  for all $i$, which means in particular that $\rho$ admits a rank-one convex decomposition in $\MM_m$. This implies that $\MM_m \subseteq \M_m$, concluding the proof that $\MM_m = \M_m$ in dimension 2 and 3. \end{proof}
Notice that the proof of the above Theorem also shows a general property of quantum states:
\begin{corollary}\label{corr:same_diagonal}
Every $\rho \in \DD$ in dimension $d \leq 3$ admits a convex decomposition into pure states as $\rho = \sum_i p_i \psi_i$ such that $\Delta(\psi_i) = \Delta(\rho)$ for all $i$.
\end{corollary}
\noindent This in particular solves an open question raised in \cite{chitambar_2016-3} concerning the additivity of the coherence of assistance $C_a$ in dimension $d=3$; we will address this point in more detail in Sec.~\ref{sec:C_a}.

The most important point about relaxing the set $\M_m$ to $\MM_m$ is that $\MM_m$ can be represented by simple linear matrix inequalities, allowing us to reduce many of the intractable optimization problems involved in computing the one-shot assisted distillable coherence to efficiently computable semidefinite programs. In particular, since the fidelity function is known to be computable with an SDP \cite{watrous_2009,watrous_2013}, we can define the semidefinite program
\begin{equation}\begin{aligned}
	\wt{F}_{A,\O}\left( \rho , m\right) \coloneqq \max_{\omega \in \MM_m} \,F(\rho, \omega)
\end{aligned}\end{equation}
and the following SDP relaxation of the one-shot assisted distillation rate:
\begin{equation}\begin{aligned}
	&\wt{C}_{A,\O}^{(1),\ve}(\rho) \coloneqq \log \max \lset m \in \NN \bar \wt{F}_{A,\O} (\rho, m) \geq 1\! -\! \ve \rset\\
	&\hspace{0.5em}= \log \max \lset m\! \in\! \NN \!\bar \! F(\rho, \omega)\! \geq\! 1\!-\!\ve,\; \omega\! \in\! \DD,\; \norm{\Delta(\omega)}{\!\infty}\! \leq\! \frac{1}{m} \rset\\
	&\hspace{0.5em}= \logfloor{ - \log \min_{\omega \in \B_\ve(\rho)} \norm{\Delta(\omega)}{\infty} }.
\end{aligned}\end{equation}
In addition to establishing a general upper bound on the one-shot distillable coherence in the assisted setting, an application of Thm.~\ref{thm:correlation_matrices} then allows us to exactly characterize of the rate of distillation for low-dimensional systems.
\begin{corollary}For every $\rho \in \DD$ it holds that $C_{A,\O}^{(1),\ve}(\rho) \leq \wt{C}_{A,\O}^{(1),\ve}(\rho)$, with equality if $d \leq 3$.\end{corollary}
As a particular case of this result, the zero-error assisted distillable coherence is given for $d\leq 3$ precisely by $C_{A,\O}^{(1),0}(\rho) = \log \left\lfloor \norm{\Delta(\rho)}{\infty}^{-1} \right\rfloor$. It is straightforward to see from Cor.~\ref{corr:same_diagonal} that, since any state $\rho$ in $d \leq 3$ admits a pure-state decomposition $\{p_{i}, \psi_{i}\}$ such that $\Delta(\psi_i) = \Delta(\rho)$ for all $i$, then so does $\rho^{\otimes n}$ for any $n$, and so the rate of assisted distillation of a many-copy state will depend only on $\norm{\Delta(\rho^{\otimes n})}{\infty} = \norm{\Delta(\rho)}{\infty}^n$. In the asymptotic limit, we therefore obtain the zero-error assisted distillation rate
\begin{equation}\begin{aligned}
	\lim_{n \to \infty} \frac{1}{n} C_{A,\O}^{(1),0}(\rho^{\otimes n}) = - \log \norm{\Delta(\rho)}{\infty}
\end{aligned}\end{equation}
for any qubit or qutrit system $\rho$. For larger dimensions of $\rho$, $\log \left\lfloor \norm{\Delta(\rho)}{\infty}^{-1} \right\rfloor$ provides an upper bound for the zero-error rate $C_{A,\O}^{(1),0}(\rho)$, tight in all dimensions for $\rho$ pure \cite{regula_2017}. Similarly, $-\log \norm{\Delta(\rho)}{\infty}$ is then an upper bound on the asymptotic zero-error assisted distillation rate, with equality for pure-state inputs.


\begin{figure*}[t]
\centering
\begin{adjustbox}{center}
\begin{subfigure}{.33\textwidth}
  \centering
  \includegraphics[width=6cm]{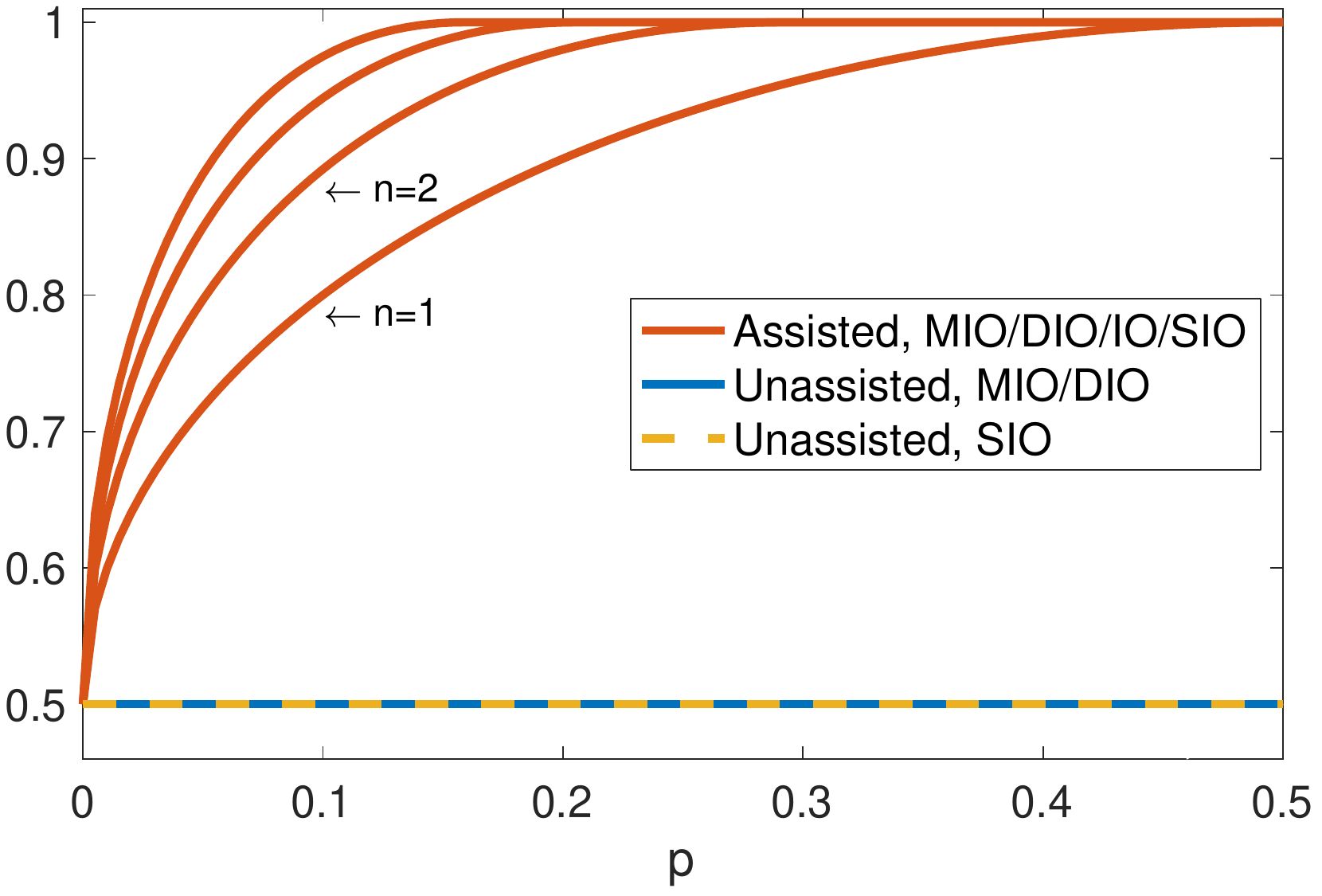}
  \caption{$\rho = p \proj{0} + (1-p) \proj{1}$\vphantom{$\left(\begin{smallmatrix}p & 2 p (1-p) \\ p (1-p) & 1-p\end{smallmatrix}\right)$}}
  \label{fig:sub1}
\end{subfigure}
\begin{subfigure}{.33\textwidth}
  \centering
  \includegraphics[width=6cm]{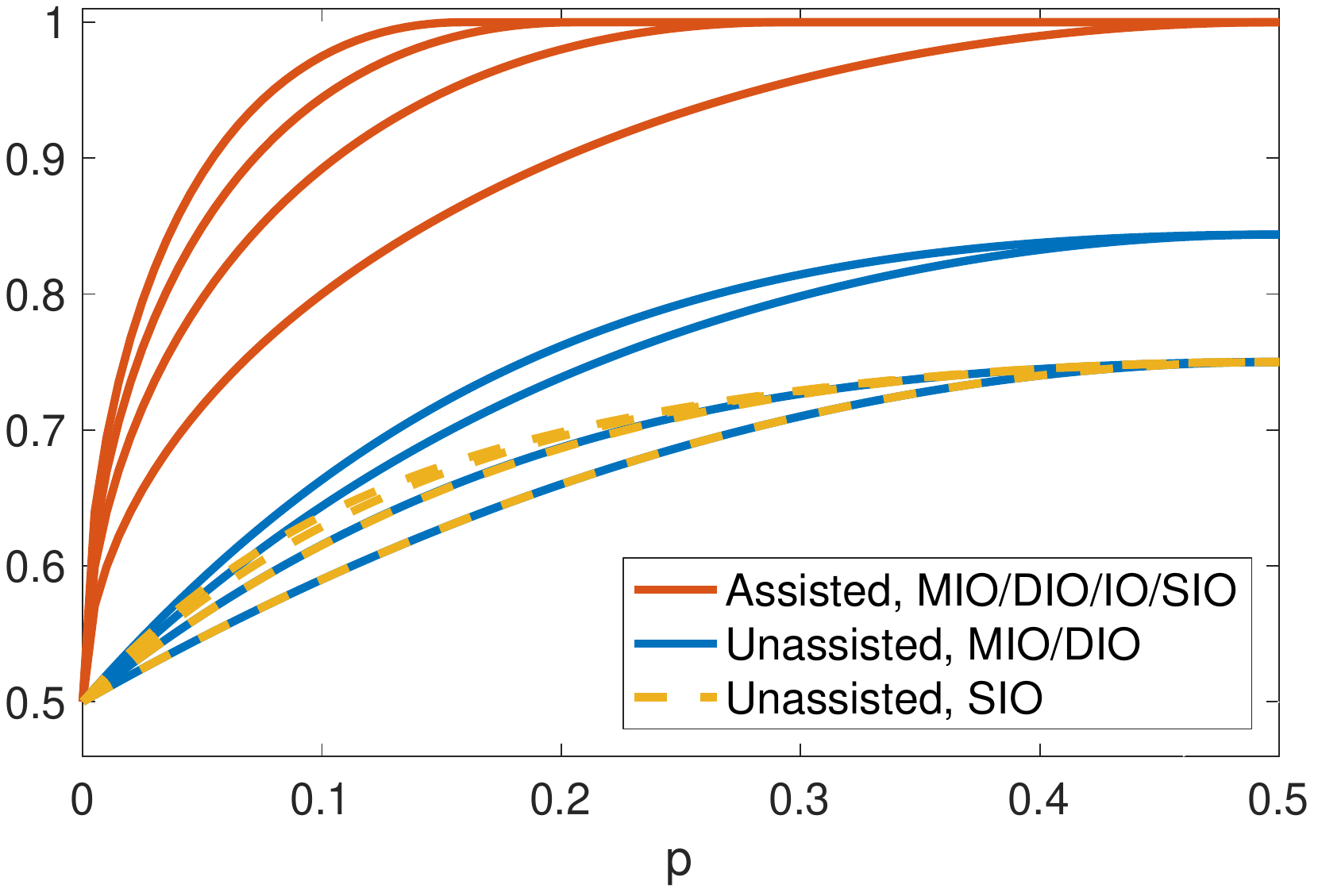}
  \caption{$\rho = \left(\begin{smallmatrix}p & p (1-p) \\ p (1-p) & 1-p\end{smallmatrix}\right)$}
  \label{fig:sub2}
\end{subfigure}
\begin{subfigure}{.33\textwidth}
  \centering
  \includegraphics[width=6cm]{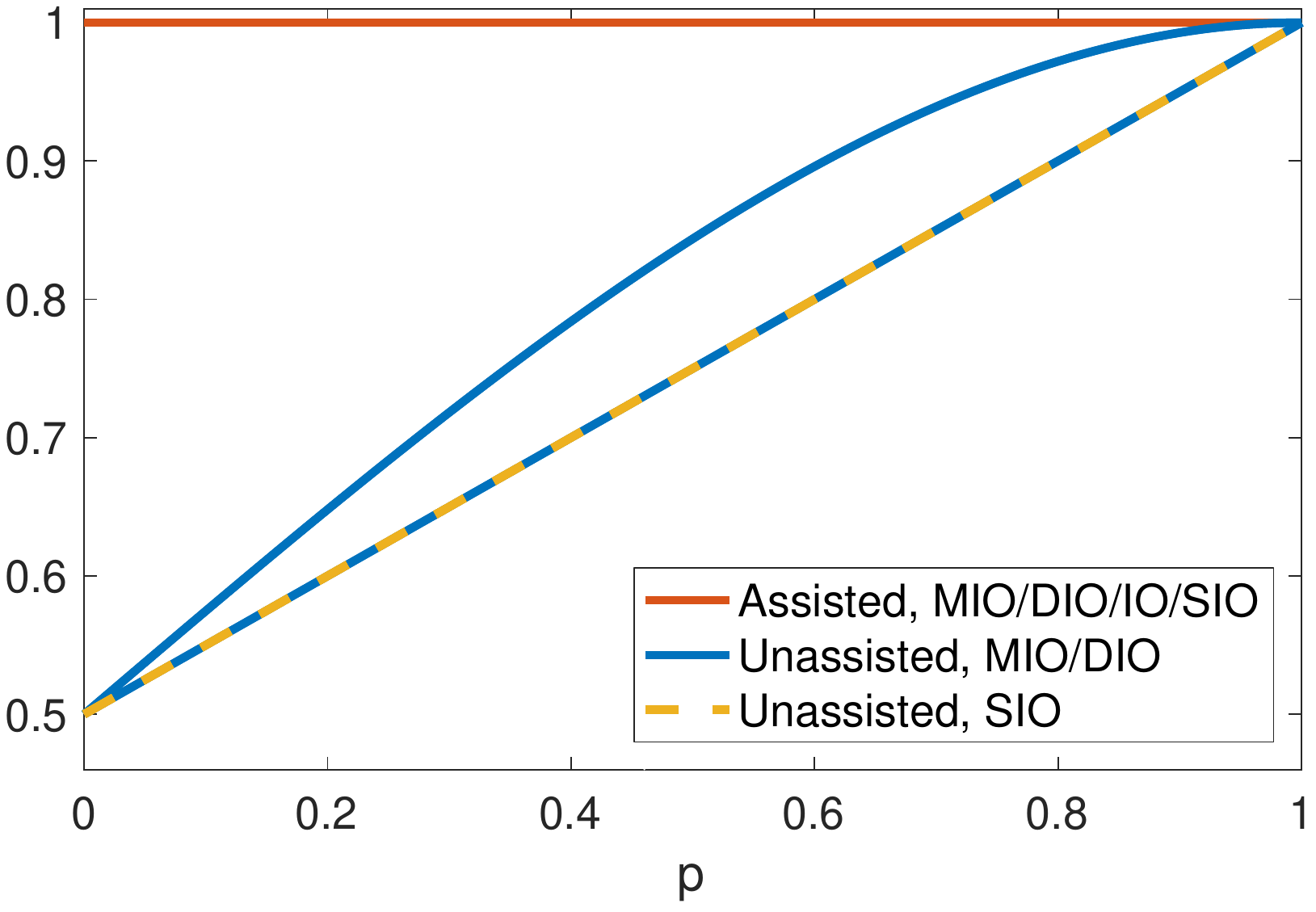}
  \caption{$\rho = p \Psi_2 + (1-p) \id/2$\vphantom{$\left(\begin{smallmatrix}p & 2 p (1-p) \\ 2 p (1-p) & 1-p\end{smallmatrix}\right)$}}
  \label{fig:sub3}
\end{subfigure}
\end{adjustbox}
\captionsetup{width=.97\linewidth,justification=raggedright}
\caption{How does increasing the number of copies affect the achievable fidelity of distillation of a single bit of coherence? The figures present a comparison between the fidelity of assisted distillation $F_{A,\O}(\rho^{\otimes n}, 2)$ as obtained in this work, as well as the fidelity of unassisted distillation $F_{\O}(\rho^{\otimes n}, 2)$ for $\O \in \{\MIO,\DIO\}$ as obtained in \cite{regula_2017} and for $\O = \SIO$ as obtained in \cite{zhao_2018-1,lami_2018-1}. In each case, we consider the distillation from the state $\rho^{\otimes n}$ for $n \in \{1,2,3,4\}$.\\
Figure (a) shows in particular the fundamental difference between assisted and unassisted distillation in that assistance allows for distillation from incoherent states. Figure (b) showcases the advantage provided by the assistance, as one can see that the state $\rho$ admits a value of $p < \frac{1}{2}$ such that perfect distillation of $\Psi_2$ is possible already from one copy of $\rho$ with assistance, while even four copies are insufficient without it. We further recall from \cite{lami_2018-1} that the SIO fidelity is bounded away from 1 for any number of copies of the state $\rho$ in (b). We finally remark the curious phenomenon in (c), where increasing the number of copies does not increase the achievable fidelity under SIO (previously noted in \cite{lami_2018-1}), while the increase in fidelity of distillation for MIO/DIO occurs only for odd number of copies $\rho^{\otimes n}$. At the same time, perfect assisted distillation is possible already for one copy of the state, since the maximally mixed state is as useful as the maximally coherent state for assisted distillation.}
\label{fig:comparison}
\end{figure*}


\section{Quantifying the fidelity of distillation}

We will now introduce an easily computable bound on the fidelity of distillation $F_{A,\O}$, and show that it is tight for all qubit and qutrit systems, leading to an analytical characterization of the distillation fidelity for any finite number of copies of a given state. To this end, let $\ket{\delta(\rho)}$ denote the $d$-dimensional vector obtained from the square roots of the diagonal elements of $\rho$, i.e.\ $\delta(\rho) = \diag(\sqrt{\Delta(\rho)})$. We can then notice the following relation:
\begin{theorem}\label{thm:fidelity}For any integer $m\geq 1$ it holds that $\begin{displaystyle}F_{A,\O} (\rho, m) \leq \frac{1}{m} \mnorm{\delta(\rho)}^2\end{displaystyle}$, with equality if $d \leq 3$ or if $\rho = \sigma^{\otimes n}$ for $n \in \mathbb{N}$ and a state $\sigma$ of dimension at most $3$.\end{theorem}
\begin{proof}
Let $\RR^d_+$ denote the set of all vectors $\ket{x}$ such that $x_i \geq 0$ for all $i$. We first make note of the fact that
\begin{equation}\begin{aligned}
	&\lsetr \ket{\delta(\omega)} \barr \omega \in \MM_m \rsetr \\
	=& \lsetr \ket{x} \barr \lnorm{\ket{x}}{\infty} \leq \frac{1}{\sqrt{m}},\; \lnorm{\ket{x}}{2} = 1,\; \ket{x} \in \RR^d_+ \rsetr.
\end{aligned}\end{equation}
This can be seen by noting that every matrix $\omega \in \MM_m$ is a normalized density matrix, which implies the conditions that $\ket{\delta(\omega)} \in \RR^d_+$ as well as $\Tr(\omega) = \lnorm{\ket{\delta(\omega)}}{2}^2 = 1$; furthermore, each such $\omega$ satisfies $\norm{\Delta(\omega)}{\infty} = \lnorm{\ket{\delta(\omega)}}{\infty}^2 \leq \frac{1}{m}$. We now use the fact that, for any vector $\ket{y} \in \RR^d_+$ s.t. $\lnorm{\ket{y}}{2} = 1$, to compute the $m$-distillation norm it suffices to optimize over vectors in $\RR^d_+$, that is 
\begin{equation}\begin{aligned}
 \mnorm{\ket{y}} = \max \lset \braket{y|\delta(\omega)} \bar \omega \in \MM_m \rset,
\end{aligned}\end{equation}
which can be seen by using the explicit form of an optimal solution of the above optimization obtained in \cite[Thm. 4]{regula_2017}. Using the monotonicity of fidelity under quantum channels, we now have
\begin{equation}\begin{aligned}
	F_{A,\O} (\rho,m) &\leq \wt{F}_{A,\O} (\rho,m)\\
	&= \max_{\omega \in \MM_m} F(\rho, \omega)\\
	&\leq \max_{\omega \in \MM_m} F(\Delta(\rho), \Delta(\omega))\\
	&= \max_{\omega \in \MM_m} \norm{\sqrt{\Delta(\rho)} \sqrt{\vphantom{\Delta(\rho)}\Delta(\omega)}}{1}^2\\
	&= \max_{\omega \in \MM_m} \left( \sum_i \delta(\rho)_i \,\delta(\omega)_i\right)^2\\
	&= \max \lsetr \braket{\delta(\rho)|\delta(\omega)}^2 \barr \omega \in \MM_m \rsetr\\
	&= \frac{1}{m} \mnorm{\delta(\rho)}^2.
\end{aligned}\end{equation}
This establishes the first part of the Theorem. To show achievability, note by Cor.~\ref{corr:same_diagonal} that any density matrix with $d \in \{2,3\}$ (or a tensor product thereof) admits a pure-state decomposition into $\rho = \sum_i p_i \psi_i$ such that $\Delta(\psi_i) = \Delta(\rho)\;\forall i$. Writing $\rho$ in this decomposition, we then have
\begin{align}
	F_{A,\O} (\rho,m) &\geq \sum_i p_i  F_\O \left({\psi_i}, m\right)\nonumber\\
	&= \sum_i p_i  \mnorm{\ket{\psi_i}}^2\\
	&= \mnorm{\ket{\psi_i}}^2\nonumber\\
	&= \mnorm{\delta(\rho)}^2\nonumber\tag*{\qedhere}.
\end{align}
\end{proof} 
Theorem \ref{thm:fidelity} allows us to compute the achievable distillation fidelity exactly for any number of copies of a qubit or qutrit system. We stress that $\mnorm{\cdot}$ admits a semi-analytical formula \cite{regula_2017}, making the evaluation of the fidelity straightforward:
\begin{equation}\begin{aligned}
	\mnorm{\delta(\rho)} = \lnorm{\delta(\rho)^\downarrow_{1:m-k^\star}}{1} + \sqrt{k^\star}\lnorm{\delta(\rho)^\downarrow_{m-k^\star+1:d}}{2} 
\end{aligned}\end{equation}
where $\delta(\rho)^\downarrow$ denotes the vector $\delta(\rho)$ with coefficients arranged in non-increasing order by magnitude, $\delta(\rho)^\downarrow_{a:b}$ refers to a subvector consisting of the corresponding range of coefficients of $\delta(\rho)^\downarrow$, and $k^\star \coloneqq \argmin_{1\leq k\leq m} \frac{1}{\sqrt{k}}\lnorm{\delta(\rho)^\downarrow_{m-k+1:d}}{2}$.
In particular, we have that
\begin{equation*}\begin{aligned}
	F_{A,\O}(\rho,2) = \begin{cases} 1, & \norm{\Delta(\rho)}{\infty} \leq \frac{1}{2}\\ \sqrt{\norm{\Delta(\rho)}{\infty}\left(1 - \norm{\Delta(\rho)}{\infty}\right)} + \frac{1}{2},\!\! &  \norm{\Delta(\rho)}{\infty} \geq \frac{1}{2} \end{cases}
\end{aligned}\end{equation*}
for any qubit or qutrit state $\rho$. In the case of $d=3$ and $m=3$, we can similarly compute the fidelity as
\begin{equation}\begin{aligned}
	F_{A,\O}(\rho,3) = \frac13 \left(\sqrt{\rho_{11}}+\sqrt{\rho_{22}}+\sqrt{\rho_{33}}\right)^2
\end{aligned}\end{equation}
where $\rho_{ii}$ denote the diagonal elements of $\rho$.

Importantly, the setting of assisted coherence distillation in which Alice and Bob share a two-qubit system has recently found application in experimental setups~\cite{wu_2017,wu_2017-1}. This immediately demonstrates the applicability of our characterization to such experimental investigations. We emphasize that in settings where only a finite number of copies of the total state is available, the one-shot fidelity $F_{A,\O}$ and rate $C_{A,\O}^{(1),\ve}$ are more meaningful than asymptotic figures of merit such as the regularized coherence of assistance.

To explicitly and quantitatively demonstrate the advantage provided by assistance in the task of non-asymptotic coherence distillation, one may wish to compare our results with the case of unassisted distillation \cite{regula_2017,zhao_2018-1,lami_2018-1}. A comparison between the achievable fidelities of distillation in the few-copy setting, including for the states considered experimentally in \cite{wu_2017}, is presented in Fig. \ref{fig:comparison}.

Additionally, based on numerical evidence, we can conjecture that $\wt{F}_{A,\O} (\rho, m) = \frac{1}{m} \mnorm{\delta(\rho)}^2$ in any dimension, which would give a clear interpretation to the considered quantity $\mnorm{\delta(\rho)}$ as the fidelity with respect to the set $\MM_m$. Note in particular that \cite{regula_2017} $\mnorm{\delta(\rho)}^2 = m \!\iff\! \norm{\Delta(\rho)}{\infty} \leq \frac{1}{m} \!\iff\! \rho \in \MM_m$.


\section{Coherence of assistance}\label{sec:C_a}

A characterization of the protocol of assisted coherence distillation in the limit of infinitely many i.i.d.\ copies was introduced in Ref. \cite{chitambar_2016-3}, where the relation between the operational quantity $C^{\infty}_{A,\IO}$ and the coherence of assistance was explored.  Recall that the coherence of assistance, defined in~\eqref{Ca}, satisfies
\begin{equation}\begin{aligned}
	C^{\infty}_a(\rho) &\coloneqq \lim_{n \to \infty} \frac{1}{n} C_a(\rho^{\otimes n}) = C^{\infty}_{A,\IO}(\rho) = S(\Delta(\rho)),
\end{aligned}\end{equation}
and the results of our work show that one can replace IO with any set of operations SIO, DIO, MIO.

In \cite{chitambar_2016-3} it was shown that $C_a(\rho) = C^{\infty}_a(\rho)$ for any qubit state $\rho$, but the problem of additivity of $C_a(\rho)$ for qutrit systems was left as an open question, later explored also in \cite{zhao_2017-1}. Using the results of our work, we answer it in the affirmative\footnote{The fact that coherence of assistance is not additive for $d>3$ has been correctly pointed out in Ref. \cite{chitambar_2016-3}, but the proof provided there has a gap.}. 

\begin{proposition}\label{prop:Ca}
$C_a(\rho) = C^{\infty}_a(\rho)$ for any state $\rho$ with $d \leq 3$, but in $d \geq 4$ there exist states such that $C^{\infty}_a > C_a(\rho)$.
\end{proposition}
\begin{proof}
By Cor.~\ref{corr:same_diagonal}, every state $\rho$ in $d \leq 3$ admits a pure-state decomposition into states with diagonal elements equal to those of $\rho$. Writing $\rho$ in this decomposition, we have
\begin{equation}\begin{aligned}
	C_a(\rho) & \geq \sum_i p_i S(\Delta(\rho)) = S(\Delta(\rho)) = C^{\infty}_a(\rho).
\end{aligned}\end{equation}
Since the converse inequality $C^{\infty}_a (\rho) \geq C_a(\rho)$ is elementary, the first part of the result follows.

When $d\geq4$, from Thm.~\ref{thm:correlation_matrices} we know that there exist states $\rho \in \MM_d$ such that $\Delta(\rho)=\frac{1}{d}\id$ but there is no ensemble of pure states $\psi_{i}$ satisfying $\rho=\sum_{i} p_{i} \psi_{i}$ and $\Delta(\psi_i) = \frac{1}{d}\id$ for all $i$. This entails that for all such ensembles
\begin{equation}\begin{aligned}
	C_{a}^{\infty}(\rho) = S\left( \Delta(\rho)\right) = \log d > \sum_i p_i S(\Delta(\psi_i))
\end{aligned}\end{equation}
as the uniform distribution is the unique maximizer of the Shannon entropy. Maximizing over all ensembles yields $C^{\infty}_a(\rho) > C_a(\rho)$, as claimed.
\end{proof}
We remark that Prop.~\ref{prop:Ca} disproves a claim in Ref. \cite{uhlmann_2010} that $C_a$ is additive for any $d$.

\section{Conclusions}

We have studied the task of assisted distillation of coherence in non-asymptotic regimes and introduced a mathematical framework for its characterization. We established an exact quantification of the maximal fidelity of distillation as well as the best achievable distillation rates in the non-asymptotic setting, deriving computable SDP and analytical results for low-dimensional systems.

One of the implications of the characterization presented herein is that the best achievable rate of assisted distillation is the same regardless of the class of operations used by Bob in the protocol.  We remark that \cite{streltsov_2017-1} considered also more general settings of asymptotic assisted distillation, for example the scenario where Alice is restricted to performing incoherent local operations, as well as the case where Alice and Bob together can perform a larger class of operations called the separable incoherent operations. Interestingly, it was shown that the asymptotic rate of assisted distillation is the same in all of these cases. Our results extend this analysis by showing that even by allowing Bob access to the maximal set of free operations MIO one still cannot improve the rate of assisted distillation, even in the non-asymptotic regime.

The results provide insight into both the operational characterization as well as the mathematical formalism of the resource theory of quantum coherence, contributing to a better understanding of this fundamental resource, particularly in practical and experimentally-relevant scenarios. Exploiting the similarities between coherence and other resource theories such as thermodynamics and entanglement, we hope that our framework can find use beyond the theory of quantum coherence, as well as in generalized settings of assisted distillation.\\[.5em]

\noindent\textit{Acknowledgments.}--- We are grateful to Gerardo Adesso, Eric Chitambar, Kun Fang, Min-Hsiu Hsieh, Jamie Sikora, and Xin Wang for useful discussions. B.R.\ and L.L.\ acknowledge financial support from the European Research Council (ERC) under the Starting Grant GQCOP (Grant No.~637352). A.S.\ acknowledges financial support by the National
Science Center in Poland (POLONEZ UMO-2016/21/P/ST2/04054) and the
European Union's Horizon 2020 research and innovation programme
under the Marie Sk\l{}odowska-Curie grant agreement No. 665778. \\[.5em]
\textit{Note.}--- During the completion of this work, an independent study of one-shot assisted distillation of coherence under IO was reported in \cite{vijayan_2018} using a different set of methods.

\bibliographystyle{apsrev4-1}
\bibliography{main}

\end{document}

%% file: def.tex
\usepackage{amsthm,amsmath,amssymb,amsbsy}

\usepackage{color}
\usepackage[dvipsnames]{xcolor}
\usepackage[colorlinks=true,allcolors=MidnightBlue]{hyperref}
\usepackage{url}

\usepackage{graphicx}
\usepackage{bm,bbm}%
\usepackage{braket}
\usepackage{enumerate}
\usepackage{subcaption,caption}
\usepackage{booktabs,multirow}
\usepackage{mathtools,bbding}
\usepackage{microtype}
\usepackage{mathrsfs}

\usepackage{txfonts,newtxmath}

\usepackage[most,breakable]{tcolorbox}

\newcommand{\nc}{\newcommand}
\newcommand{\rnc}{\renewcommand}

\rnc{\thesection}{\arabic{section}}
\rnc{\thesubsection}{\thesection.\arabic{subsection}}
\rnc{\thesubsubsection}{\thesubsection.\arabic{subsubsection}}

\usepackage{etoolbox}
\hypersetup{linktoc=all}

\newtheorem{definition}{Definition}
\newtheorem{proposition}[definition]{Proposition}
\newtheorem{lemma}[definition]{Lemma}

\newtheorem{theorem}[definition]{Theorem}
\newtheorem{corollary}[definition]{Corollary}

\newtheoremstyle{red}{}{}{\itshape}{}{\color{red!80!black}\bfseries}{.}{ }{}
\theoremstyle{red}

\DeclareMathOperator{\Tr}{Tr}

\DeclareMathOperator{\diag}{diag}
\DeclareMathOperator{\conv}{conv}

\nc{\psic}{\psi^{c}}
\nc{\two}[1]{\underline{2^{d-#1}}}
\rnc{\H}{\mathcal{H}}
\nc{\Hanc}{\mathcal{H}_{\text{anc}}}
\nc{\psianc}{\psi_{\text{anc}}}
\nc{\lampow}{\lambda^{1/d}}

\nc{\norm}[2]{\left\lVert#1\right\rVert_{\,#2}}
\nc{\proj}[1]{\ket{#1}\!\bra{#1}}
\nc{\pro}[1]{#1 #1^\dagger}
\nc{\lnorm}[2]{\left\lVert#1\right\rVert_{\ell_{#2}}}

\nc{\RR}{{{\mathbb R}}}
\nc{\CC}{{{\mathbb C}}}
\nc{\FF}{{{\mathbb F}}}
\nc{\NN}{{{\mathbb N}}}
\nc{\ZZ}{{{\mathbb Z}}}

\nc{\MIO}{{\text{\rm MIO}}}
\nc{\DIO}{{\text{\rm DIO}}}
\nc{\SIO}{{\text{\rm SIO}}}
\nc{\IO}{{\text{\rm IO}}}

\nc{\SEP}{{\text{SEP}}}
\nc{\NS}{{\text{NS}}}
\nc{\LOCC}{{\text{LOCC}}}
\nc{\PPT}{{\text{PPT}}}
\nc{\EXT}{{\text{EXT}}}
\nc{\OLOCC}{{\text{1-LOCC}}}
\nc{\SEPP}{{\text{SEPP}}}

\nc{\MC}{{\text{\rm MC}}}

\nc{\cE}{\mathscr{E}}

\rnc{\*}{\textup{*}}
\nc{\<}{\left\langle}
\rnc{\>}{\right\rangle}

\rnc{\bar}{\;\rule{0pt}{9.5pt}\right|\;}

\nc{\lset}{\left\{\left.}
\nc{\rset}{\right\}}
\nc{\lsetr}{\left\{\,}
\nc{\rsetr}{\right.\right\}}
\nc{\barr}{\,\rule{0pt}{9.5pt}\left|\;}
\nc{\ketbra}[2]{\ket{#1}\!\bra{#2}}

\DeclareMathOperator*{\argmin}{arg\,min}

\nc{\logfloor}[1]{\left\lfloor {#1} \right\rfloor_{\log}}

\newcommand{\DD}{\mathbb{D}}

\newcommand{\cbraket}[1]{\left|\braket{#1}\right|}

\newcommand{\id}{\mathbbm{1}}

\newcommand{\B}{\mathcal{B}}

\renewcommand{\O}{\mathcal{O}}
\newcommand{\I}{\mathcal{I}}
\nc{\MM}{\widetilde{\M}}
\nc{\Ml}{\M^{\leq}}
\def\ve{\varepsilon}

\nc{\wt}{\widetilde}

\nc{\SDP}{\text{\rm SDP}}

\nc{\cc}{{\circ\circ}}
\nc{\mnorm}[1]{\norm{#1}{[m]}}

\nc{\F}{\mathcal{F}}
\nc{\M}{\mathcal{M}}

\let\oldproofname\proofname
\rnc{\proofname}{\rm\bf{\oldproofname}}
\rnc{\qedsymbol}{{\color{gray!50!black}\rule{0.6em}{0.6em}}}

%% file: assisted-pra-resub.bbl
\begin{thebibliography}{58}%
\makeatletter
\providecommand \@ifxundefined [1]{%
 \@ifx{#1\undefined}
}%
\providecommand \@ifnum [1]{%
 \ifnum #1\expandafter \@firstoftwo
 \else \expandafter \@secondoftwo
 \fi
}%
\providecommand \@ifx [1]{%
 \ifx #1\expandafter \@firstoftwo
 \else \expandafter \@secondoftwo
 \fi
}%
\providecommand \natexlab [1]{#1}%
\providecommand \enquote  [1]{``#1''}%
\providecommand \bibnamefont  [1]{#1}%
\providecommand \bibfnamefont [1]{#1}%
\providecommand \citenamefont [1]{#1}%
\providecommand \href@noop [0]{\@secondoftwo}%
\providecommand \href [0]{\begingroup \@sanitize@url \@href}%
\providecommand \@href[1]{\@@startlink{#1}\@@href}%
\providecommand \@@href[1]{\endgroup#1\@@endlink}%
\providecommand \@sanitize@url [0]{\catcode `\\12\catcode `\$12\catcode
  `\&12\catcode `\#12\catcode `\^12\catcode `\_12\catcode `\%12\relax}%
\providecommand \@@startlink[1]{}%
\providecommand \@@endlink[0]{}%
\providecommand \url  [0]{\begingroup\@sanitize@url \@url }%
\providecommand \@url [1]{\endgroup\@href {#1}{\urlprefix }}%
\providecommand \urlprefix  [0]{URL }%
\providecommand \Eprint [0]{\href }%
\providecommand \doibase [0]{http://dx.doi.org/}%
\providecommand \selectlanguage [0]{\@gobble}%
\providecommand \bibinfo  [0]{\@secondoftwo}%
\providecommand \bibfield  [0]{\@secondoftwo}%
\providecommand \translation [1]{[#1]}%
\providecommand \BibitemOpen [0]{}%
\providecommand \bibitemStop [0]{}%
\providecommand \bibitemNoStop [0]{.\EOS\space}%
\providecommand \EOS [0]{\spacefactor3000\relax}%
\providecommand \BibitemShut  [1]{\csname bibitem#1\endcsname}%
\let\auto@bib@innerbib\@empty
\bibitem [{\citenamefont {Bennett}\ \emph
  {et~al.}(1996{\natexlab{a}})\citenamefont {Bennett}, \citenamefont
  {Bernstein}, \citenamefont {Popescu},\ and\ \citenamefont
  {Schumacher}}]{bennett_1996-1}%
  \BibitemOpen
  \bibfield  {author} {\bibinfo {author} {\bibfnamefont {C.~H.}\ \bibnamefont
  {Bennett}}, \bibinfo {author} {\bibfnamefont {H.~J.}\ \bibnamefont
  {Bernstein}}, \bibinfo {author} {\bibfnamefont {S.}~\bibnamefont {Popescu}},
  \ and\ \bibinfo {author} {\bibfnamefont {B.}~\bibnamefont {Schumacher}},\
  }\href {\doibase 10.1103/PhysRevA.53.2046} {\bibfield  {journal} {\bibinfo
  {journal} {Phys. Rev. A}\ }\textbf {\bibinfo {volume} {53}},\ \bibinfo
  {pages} {2046} (\bibinfo {year} {1996}{\natexlab{a}})}\BibitemShut {NoStop}%
\bibitem [{\citenamefont {Bennett}\ \emph
  {et~al.}(1996{\natexlab{b}})\citenamefont {Bennett}, \citenamefont
  {Brassard}, \citenamefont {Popescu}, \citenamefont {Schumacher},
  \citenamefont {Smolin},\ and\ \citenamefont {Wootters}}]{bennett_1996-3}%
  \BibitemOpen
  \bibfield  {author} {\bibinfo {author} {\bibfnamefont {C.~H.}\ \bibnamefont
  {Bennett}}, \bibinfo {author} {\bibfnamefont {G.}~\bibnamefont {Brassard}},
  \bibinfo {author} {\bibfnamefont {S.}~\bibnamefont {Popescu}}, \bibinfo
  {author} {\bibfnamefont {B.}~\bibnamefont {Schumacher}}, \bibinfo {author}
  {\bibfnamefont {J.~A.}\ \bibnamefont {Smolin}}, \ and\ \bibinfo {author}
  {\bibfnamefont {W.~K.}\ \bibnamefont {Wootters}},\ }\href {\doibase
  10.1103/PhysRevLett.76.722} {\bibfield  {journal} {\bibinfo  {journal} {Phys.
  Rev. Lett.}\ }\textbf {\bibinfo {volume} {76}},\ \bibinfo {pages} {722}
  (\bibinfo {year} {1996}{\natexlab{b}})}\BibitemShut {NoStop}%
\bibitem [{\citenamefont {Rains}(1999)}]{rains_1999}%
  \BibitemOpen
  \bibfield  {author} {\bibinfo {author} {\bibfnamefont {E.~M.}\ \bibnamefont
  {Rains}},\ }\href {\doibase 10.1103/PhysRevA.60.173} {\bibfield  {journal}
  {\bibinfo  {journal} {Phys. Rev. A}\ }\textbf {\bibinfo {volume} {60}},\
  \bibinfo {pages} {173} (\bibinfo {year} {1999})}\BibitemShut {NoStop}%
\bibitem [{\citenamefont {DiVincenzo}\ \emph {et~al.}(1999)\citenamefont
  {DiVincenzo}, \citenamefont {Fuchs}, \citenamefont {Mabuchi}, \citenamefont
  {Smolin}, \citenamefont {Thapliyal},\ and\ \citenamefont
  {Uhlmann}}]{divincenzo_1999}%
  \BibitemOpen
  \bibfield  {author} {\bibinfo {author} {\bibfnamefont {D.~P.}\ \bibnamefont
  {DiVincenzo}}, \bibinfo {author} {\bibfnamefont {C.~A.}\ \bibnamefont
  {Fuchs}}, \bibinfo {author} {\bibfnamefont {H.}~\bibnamefont {Mabuchi}},
  \bibinfo {author} {\bibfnamefont {J.~A.}\ \bibnamefont {Smolin}}, \bibinfo
  {author} {\bibfnamefont {A.}~\bibnamefont {Thapliyal}}, \ and\ \bibinfo
  {author} {\bibfnamefont {A.}~\bibnamefont {Uhlmann}},\ }in\ \href {\doibase
  10.1007/3-540-49208-9_21} {\emph {\bibinfo {booktitle} {Quantum {{Computing}}
  and {{Quantum Communications}}}}},\ \bibinfo {series and number} {Lecture
  Notes in Computer Science}\ (\bibinfo  {publisher} {{Springer, Berlin,
  Heidelberg}},\ \bibinfo {year} {1999})\ pp.\ \bibinfo {pages}
  {247--257}\BibitemShut {NoStop}%
\bibitem [{\citenamefont {Smolin}\ \emph {et~al.}(2005)\citenamefont {Smolin},
  \citenamefont {Verstraete},\ and\ \citenamefont {Winter}}]{smolin_2005}%
  \BibitemOpen
  \bibfield  {author} {\bibinfo {author} {\bibfnamefont {J.~A.}\ \bibnamefont
  {Smolin}}, \bibinfo {author} {\bibfnamefont {F.}~\bibnamefont {Verstraete}},
  \ and\ \bibinfo {author} {\bibfnamefont {A.}~\bibnamefont {Winter}},\ }\href
  {\doibase 10.1103/PhysRevA.72.052317} {\bibfield  {journal} {\bibinfo
  {journal} {Phys. Rev. A}\ }\textbf {\bibinfo {volume} {72}},\ \bibinfo
  {pages} {052317} (\bibinfo {year} {2005})}\BibitemShut {NoStop}%
\bibitem [{\citenamefont {Rains}(2001)}]{rains_2001}%
  \BibitemOpen
  \bibfield  {author} {\bibinfo {author} {\bibfnamefont {E.~M.}\ \bibnamefont
  {Rains}},\ }\href {\doibase 10.1109/18.959270} {\bibfield  {journal}
  {\bibinfo  {journal} {IEEE Trans. Inf. Theory}\ }\textbf {\bibinfo {volume}
  {47}},\ \bibinfo {pages} {2921} (\bibinfo {year} {2001})}\BibitemShut
  {NoStop}%
\bibitem [{\citenamefont {Brand{\~a}o}\ and\ \citenamefont
  {Datta}(2011)}]{brandao_2011}%
  \BibitemOpen
  \bibfield  {author} {\bibinfo {author} {\bibfnamefont {F.~G. S.~L.}\
  \bibnamefont {Brand{\~a}o}}\ and\ \bibinfo {author} {\bibfnamefont
  {N.}~\bibnamefont {Datta}},\ }\href {\doibase 10.1109/TIT.2011.2104531}
  {\bibfield  {journal} {\bibinfo  {journal} {IEEE Trans. Inf. Theory}\
  }\textbf {\bibinfo {volume} {57}},\ \bibinfo {pages} {1754} (\bibinfo {year}
  {2011})}\BibitemShut {NoStop}%
\bibitem [{\citenamefont {Buscemi}\ and\ \citenamefont
  {Datta}(2010)}]{buscemi_2010}%
  \BibitemOpen
  \bibfield  {author} {\bibinfo {author} {\bibfnamefont {F.}~\bibnamefont
  {Buscemi}}\ and\ \bibinfo {author} {\bibfnamefont {N.}~\bibnamefont
  {Datta}},\ }\href {\doibase 10.1109/TIT.2009.2039166} {\bibfield  {journal}
  {\bibinfo  {journal} {IEEE Trans. Inf. Theory}\ }\textbf {\bibinfo {volume}
  {56}},\ \bibinfo {pages} {1447} (\bibinfo {year} {2010})}\BibitemShut
  {NoStop}%
\bibitem [{\citenamefont {Buscemi}\ and\ \citenamefont
  {Datta}(2013)}]{buscemi_2013}%
  \BibitemOpen
  \bibfield  {author} {\bibinfo {author} {\bibfnamefont {F.}~\bibnamefont
  {Buscemi}}\ and\ \bibinfo {author} {\bibfnamefont {N.}~\bibnamefont
  {Datta}},\ }\href {\doibase 10.1109/TIT.2012.2227673} {\bibfield  {journal}
  {\bibinfo  {journal} {IEEE Trans. Inf. Theory}\ }\textbf {\bibinfo {volume}
  {59}},\ \bibinfo {pages} {1940} (\bibinfo {year} {2013})}\BibitemShut
  {NoStop}%
\bibitem [{\citenamefont {Leditzky}\ \emph {et~al.}(2018)\citenamefont
  {Leditzky}, \citenamefont {Datta},\ and\ \citenamefont
  {Smith}}]{leditzky_2017}%
  \BibitemOpen
  \bibfield  {author} {\bibinfo {author} {\bibfnamefont {F.}~\bibnamefont
  {Leditzky}}, \bibinfo {author} {\bibfnamefont {N.}~\bibnamefont {Datta}}, \
  and\ \bibinfo {author} {\bibfnamefont {G.}~\bibnamefont {Smith}},\ }\href
  {\doibase 10.1109/TIT.2017.2776907} {\bibfield  {journal} {\bibinfo
  {journal} {IEEE Trans. Inf. Theory}\ }\textbf {\bibinfo {volume} {64}},\
  \bibinfo {pages} {4689} (\bibinfo {year} {2018})}\BibitemShut {NoStop}%
\bibitem [{\citenamefont {Fang}\ \emph {et~al.}(2017)\citenamefont {Fang},
  \citenamefont {Wang}, \citenamefont {Tomamichel},\ and\ \citenamefont
  {Duan}}]{fang_2017}%
  \BibitemOpen
  \bibfield  {author} {\bibinfo {author} {\bibfnamefont {K.}~\bibnamefont
  {Fang}}, \bibinfo {author} {\bibfnamefont {X.}~\bibnamefont {Wang}}, \bibinfo
  {author} {\bibfnamefont {M.}~\bibnamefont {Tomamichel}}, \ and\ \bibinfo
  {author} {\bibfnamefont {R.}~\bibnamefont {Duan}},\ }\href@noop {} {\
  (\bibinfo {year} {2017})},\ \Eprint {http://arxiv.org/abs/1706.06221}
  {arXiv:1706.06221} \BibitemShut {NoStop}%
\bibitem [{\citenamefont {Horodecki}\ \emph {et~al.}(2009)\citenamefont
  {Horodecki}, \citenamefont {Horodecki}, \citenamefont {Horodecki},\ and\
  \citenamefont {Horodecki}}]{horodecki_2009}%
  \BibitemOpen
  \bibfield  {author} {\bibinfo {author} {\bibfnamefont {R.}~\bibnamefont
  {Horodecki}}, \bibinfo {author} {\bibfnamefont {P.}~\bibnamefont
  {Horodecki}}, \bibinfo {author} {\bibfnamefont {M.}~\bibnamefont
  {Horodecki}}, \ and\ \bibinfo {author} {\bibfnamefont {K.}~\bibnamefont
  {Horodecki}},\ }\href {\doibase 10.1103/RevModPhys.81.865} {\bibfield
  {journal} {\bibinfo  {journal} {Rev. Mod. Phys.}\ }\textbf {\bibinfo {volume}
  {81}},\ \bibinfo {pages} {865} (\bibinfo {year} {2009})}\BibitemShut
  {NoStop}%
\bibitem [{\citenamefont {Brand{\~a}o}\ and\ \citenamefont
  {Gour}(2015)}]{brandao_2015}%
  \BibitemOpen
  \bibfield  {author} {\bibinfo {author} {\bibfnamefont {F.~G. S.~L.}\
  \bibnamefont {Brand{\~a}o}}\ and\ \bibinfo {author} {\bibfnamefont
  {G.}~\bibnamefont {Gour}},\ }\href {\doibase 10.1103/PhysRevLett.115.070503}
  {\bibfield  {journal} {\bibinfo  {journal} {Phys. Rev. Lett.}\ }\textbf
  {\bibinfo {volume} {115}},\ \bibinfo {pages} {070503} (\bibinfo {year}
  {2015})}\BibitemShut {NoStop}%
\bibitem [{\citenamefont {Lami}\ \emph
  {et~al.}(2018{\natexlab{a}})\citenamefont {Lami}, \citenamefont {Regula},
  \citenamefont {Wang}, \citenamefont {Nichols}, \citenamefont {Winter},\ and\
  \citenamefont {Adesso}}]{lami_2018}%
  \BibitemOpen
  \bibfield  {author} {\bibinfo {author} {\bibfnamefont {L.}~\bibnamefont
  {Lami}}, \bibinfo {author} {\bibfnamefont {B.}~\bibnamefont {Regula}},
  \bibinfo {author} {\bibfnamefont {X.}~\bibnamefont {Wang}}, \bibinfo {author}
  {\bibfnamefont {R.}~\bibnamefont {Nichols}}, \bibinfo {author} {\bibfnamefont
  {A.}~\bibnamefont {Winter}}, \ and\ \bibinfo {author} {\bibfnamefont
  {G.}~\bibnamefont {Adesso}},\ }\href {\doibase 10.1103/PhysRevA.98.022335}
  {\bibfield  {journal} {\bibinfo  {journal} {Phys. Rev. A}\ }\textbf {\bibinfo
  {volume} {98}},\ \bibinfo {pages} {022335} (\bibinfo {year}
  {2018}{\natexlab{a}})}\BibitemShut {NoStop}%
\bibitem [{\citenamefont {Winter}\ and\ \citenamefont
  {Yang}(2016)}]{winter_2016}%
  \BibitemOpen
  \bibfield  {author} {\bibinfo {author} {\bibfnamefont {A.}~\bibnamefont
  {Winter}}\ and\ \bibinfo {author} {\bibfnamefont {D.}~\bibnamefont {Yang}},\
  }\href {\doibase 10.1103/PhysRevLett.116.120404} {\bibfield  {journal}
  {\bibinfo  {journal} {Phys. Rev. Lett.}\ }\textbf {\bibinfo {volume} {116}},\
  \bibinfo {pages} {120404} (\bibinfo {year} {2016})}\BibitemShut {NoStop}%
\bibitem [{\citenamefont {Chitambar}\ \emph {et~al.}(2016)\citenamefont
  {Chitambar}, \citenamefont {Streltsov}, \citenamefont {Rana}, \citenamefont
  {Bera}, \citenamefont {Adesso},\ and\ \citenamefont
  {Lewenstein}}]{chitambar_2016-3}%
  \BibitemOpen
  \bibfield  {author} {\bibinfo {author} {\bibfnamefont {E.}~\bibnamefont
  {Chitambar}}, \bibinfo {author} {\bibfnamefont {A.}~\bibnamefont
  {Streltsov}}, \bibinfo {author} {\bibfnamefont {S.}~\bibnamefont {Rana}},
  \bibinfo {author} {\bibfnamefont {M.}~\bibnamefont {Bera}}, \bibinfo {author}
  {\bibfnamefont {G.}~\bibnamefont {Adesso}}, \ and\ \bibinfo {author}
  {\bibfnamefont {M.}~\bibnamefont {Lewenstein}},\ }\href {\doibase
  10.1103/PhysRevLett.116.070402} {\bibfield  {journal} {\bibinfo  {journal}
  {Phys. Rev. Lett.}\ }\textbf {\bibinfo {volume} {116}},\ \bibinfo {pages}
  {070402} (\bibinfo {year} {2016})}\BibitemShut {NoStop}%
\bibitem [{\citenamefont {Zhao}\ \emph
  {et~al.}(2018{\natexlab{a}})\citenamefont {Zhao}, \citenamefont {Liu},
  \citenamefont {Yuan}, \citenamefont {Chitambar},\ and\ \citenamefont
  {Ma}}]{zhao_2018}%
  \BibitemOpen
  \bibfield  {author} {\bibinfo {author} {\bibfnamefont {Q.}~\bibnamefont
  {Zhao}}, \bibinfo {author} {\bibfnamefont {Y.}~\bibnamefont {Liu}}, \bibinfo
  {author} {\bibfnamefont {X.}~\bibnamefont {Yuan}}, \bibinfo {author}
  {\bibfnamefont {E.}~\bibnamefont {Chitambar}}, \ and\ \bibinfo {author}
  {\bibfnamefont {X.}~\bibnamefont {Ma}},\ }\href {\doibase
  10.1103/PhysRevLett.120.070403} {\bibfield  {journal} {\bibinfo  {journal}
  {Phys. Rev. Lett.}\ }\textbf {\bibinfo {volume} {120}},\ \bibinfo {pages}
  {070403} (\bibinfo {year} {2018}{\natexlab{a}})}\BibitemShut {NoStop}%
\bibitem [{\citenamefont {Regula}\ \emph {et~al.}(2018)\citenamefont {Regula},
  \citenamefont {Fang}, \citenamefont {Wang},\ and\ \citenamefont
  {Adesso}}]{regula_2017}%
  \BibitemOpen
  \bibfield  {author} {\bibinfo {author} {\bibfnamefont {B.}~\bibnamefont
  {Regula}}, \bibinfo {author} {\bibfnamefont {K.}~\bibnamefont {Fang}},
  \bibinfo {author} {\bibfnamefont {X.}~\bibnamefont {Wang}}, \ and\ \bibinfo
  {author} {\bibfnamefont {G.}~\bibnamefont {Adesso}},\ }\href {\doibase
  10.1103/PhysRevLett.121.010401} {\bibfield  {journal} {\bibinfo  {journal}
  {Phys. Rev. Lett.}\ }\textbf {\bibinfo {volume} {121}},\ \bibinfo {pages}
  {010401} (\bibinfo {year} {2018})}\BibitemShut {NoStop}%
\bibitem [{\citenamefont {Fang}\ \emph {et~al.}(2018)\citenamefont {Fang},
  \citenamefont {Wang}, \citenamefont {Lami}, \citenamefont {Regula},\ and\
  \citenamefont {Adesso}}]{fang_2018}%
  \BibitemOpen
  \bibfield  {author} {\bibinfo {author} {\bibfnamefont {K.}~\bibnamefont
  {Fang}}, \bibinfo {author} {\bibfnamefont {X.}~\bibnamefont {Wang}}, \bibinfo
  {author} {\bibfnamefont {L.}~\bibnamefont {Lami}}, \bibinfo {author}
  {\bibfnamefont {B.}~\bibnamefont {Regula}}, \ and\ \bibinfo {author}
  {\bibfnamefont {G.}~\bibnamefont {Adesso}},\ }\href {\doibase
  10.1103/PhysRevLett.121.070404} {\bibfield  {journal} {\bibinfo  {journal}
  {Phys. Rev. Lett.}\ }\textbf {\bibinfo {volume} {121}},\ \bibinfo {pages}
  {070404} (\bibinfo {year} {2018})}\BibitemShut {NoStop}%
\bibitem [{\citenamefont {Zhao}\ \emph
  {et~al.}(2018{\natexlab{b}})\citenamefont {Zhao}, \citenamefont {Liu},
  \citenamefont {Yuan}, \citenamefont {Chitambar},\ and\ \citenamefont
  {Winter}}]{zhao_2018-1}%
  \BibitemOpen
  \bibfield  {author} {\bibinfo {author} {\bibfnamefont {Q.}~\bibnamefont
  {Zhao}}, \bibinfo {author} {\bibfnamefont {Y.}~\bibnamefont {Liu}}, \bibinfo
  {author} {\bibfnamefont {X.}~\bibnamefont {Yuan}}, \bibinfo {author}
  {\bibfnamefont {E.}~\bibnamefont {Chitambar}}, \ and\ \bibinfo {author}
  {\bibfnamefont {A.}~\bibnamefont {Winter}},\ }\href
  {http://arxiv.org/abs/1808.01885} {\  (\bibinfo {year}
  {2018}{\natexlab{b}})},\ \Eprint {http://arxiv.org/abs/1808.01885}
  {arXiv:1808.01885} \BibitemShut {NoStop}%
\bibitem [{\citenamefont {Lami}\ \emph
  {et~al.}(2018{\natexlab{b}})\citenamefont {Lami}, \citenamefont {Regula},\
  and\ \citenamefont {Adesso}}]{lami_2018-1}%
  \BibitemOpen
  \bibfield  {author} {\bibinfo {author} {\bibfnamefont {L.}~\bibnamefont
  {Lami}}, \bibinfo {author} {\bibfnamefont {B.}~\bibnamefont {Regula}}, \ and\
  \bibinfo {author} {\bibfnamefont {G.}~\bibnamefont {Adesso}},\ }\href
  {http://arxiv.org/abs/1809.06880} {\  (\bibinfo {year}
  {2018}{\natexlab{b}})},\ \Eprint {http://arxiv.org/abs/1809.06880}
  {arXiv:1809.06880} \BibitemShut {NoStop}%
\bibitem [{\citenamefont {Aberg}(2006)}]{aberg_2006}%
  \BibitemOpen
  \bibfield  {author} {\bibinfo {author} {\bibfnamefont {J.}~\bibnamefont
  {Aberg}},\ }\href {http://arxiv.org/abs/quant-ph/0612146} {\  (\bibinfo
  {year} {2006})},\ \Eprint {http://arxiv.org/abs/quant-ph/0612146}
  {arXiv:quant-ph/0612146} \BibitemShut {NoStop}%
\bibitem [{\citenamefont {Gour}\ and\ \citenamefont
  {Spekkens}(2008)}]{gour_2008}%
  \BibitemOpen
  \bibfield  {author} {\bibinfo {author} {\bibfnamefont {G.}~\bibnamefont
  {Gour}}\ and\ \bibinfo {author} {\bibfnamefont {R.~W.}\ \bibnamefont
  {Spekkens}},\ }\href {\doibase 10.1088/1367-2630/10/3/033023} {\bibfield
  {journal} {\bibinfo  {journal} {New J. Phys.}\ }\textbf {\bibinfo {volume}
  {10}},\ \bibinfo {pages} {033023} (\bibinfo {year} {2008})}\BibitemShut
  {NoStop}%
\bibitem [{\citenamefont {Baumgratz}\ \emph {et~al.}(2014)\citenamefont
  {Baumgratz}, \citenamefont {Cramer},\ and\ \citenamefont
  {Plenio}}]{baumgratz_2014}%
  \BibitemOpen
  \bibfield  {author} {\bibinfo {author} {\bibfnamefont {T.}~\bibnamefont
  {Baumgratz}}, \bibinfo {author} {\bibfnamefont {M.}~\bibnamefont {Cramer}}, \
  and\ \bibinfo {author} {\bibfnamefont {M.~B.}\ \bibnamefont {Plenio}},\
  }\href {\doibase 10.1103/PhysRevLett.113.140401} {\bibfield  {journal}
  {\bibinfo  {journal} {Phys. Rev. Lett.}\ }\textbf {\bibinfo {volume} {113}},\
  \bibinfo {pages} {140401} (\bibinfo {year} {2014})}\BibitemShut {NoStop}%
\bibitem [{\citenamefont {Streltsov}\ \emph
  {et~al.}(2017{\natexlab{a}})\citenamefont {Streltsov}, \citenamefont
  {Adesso},\ and\ \citenamefont {Plenio}}]{streltsov_2017}%
  \BibitemOpen
  \bibfield  {author} {\bibinfo {author} {\bibfnamefont {A.}~\bibnamefont
  {Streltsov}}, \bibinfo {author} {\bibfnamefont {G.}~\bibnamefont {Adesso}}, \
  and\ \bibinfo {author} {\bibfnamefont {M.~B.}\ \bibnamefont {Plenio}},\
  }\href {\doibase 10.1103/RevModPhys.89.041003} {\bibfield  {journal}
  {\bibinfo  {journal} {Rev. Mod. Phys.}\ }\textbf {\bibinfo {volume} {89}},\
  \bibinfo {pages} {041003} (\bibinfo {year} {2017}{\natexlab{a}})}\BibitemShut
  {NoStop}%
\bibitem [{\citenamefont {Gregoratti}\ and\ \citenamefont
  {Werner}(2003)}]{gregoratti_2003}%
  \BibitemOpen
  \bibfield  {author} {\bibinfo {author} {\bibfnamefont {M.}~\bibnamefont
  {Gregoratti}}\ and\ \bibinfo {author} {\bibfnamefont {R.~F.}\ \bibnamefont
  {Werner}},\ }\href {\doibase 10.1080/09500340308234541} {\bibfield  {journal}
  {\bibinfo  {journal} {J. Mod. Opt.}\ }\textbf {\bibinfo {volume} {50}},\
  \bibinfo {pages} {915} (\bibinfo {year} {2003})}\BibitemShut {NoStop}%
\bibitem [{\citenamefont {Hayden}\ and\ \citenamefont
  {King}(2005)}]{hayden_2005}%
  \BibitemOpen
  \bibfield  {author} {\bibinfo {author} {\bibfnamefont {P.}~\bibnamefont
  {Hayden}}\ and\ \bibinfo {author} {\bibfnamefont {C.}~\bibnamefont {King}},\
  }\href {http://arxiv.org/abs/quant-ph/0409026} {\bibfield  {journal}
  {\bibinfo  {journal} {Quant. Inf. Comput.}\ }\textbf {\bibinfo {volume}
  {5}},\ \bibinfo {pages} {156} (\bibinfo {year} {2005})},\ \Eprint
  {http://arxiv.org/abs/quant-ph/0409026} {arXiv:quant-ph/0409026} \BibitemShut
  {NoStop}%
\bibitem [{\citenamefont {Winter}(2007)}]{winter_2007}%
  \BibitemOpen
  \bibfield  {author} {\bibinfo {author} {\bibfnamefont {A.}~\bibnamefont
  {Winter}},\ }\href {http://arxiv.org/abs/quant-ph/0507045} {\bibfield
  {journal} {\bibinfo  {journal} {Markov Proc. Rel. Fld.}\ ,\ \bibinfo {pages}
  {297}} (\bibinfo {year} {2007})},\ \Eprint
  {http://arxiv.org/abs/quant-ph/0507045} {arXiv:quant-ph/0507045} \BibitemShut
  {NoStop}%
\bibitem [{\citenamefont {Chitambar}\ and\ \citenamefont
  {Hsieh}(2016)}]{chitambar_2016-2}%
  \BibitemOpen
  \bibfield  {author} {\bibinfo {author} {\bibfnamefont {E.}~\bibnamefont
  {Chitambar}}\ and\ \bibinfo {author} {\bibfnamefont {M.-H.}\ \bibnamefont
  {Hsieh}},\ }\href {\doibase 10.1103/PhysRevLett.117.020402} {\bibfield
  {journal} {\bibinfo  {journal} {Phys. Rev. Lett.}\ }\textbf {\bibinfo
  {volume} {117}},\ \bibinfo {pages} {020402} (\bibinfo {year}
  {2016})}\BibitemShut {NoStop}%
\bibitem [{\citenamefont {Streltsov}\ \emph
  {et~al.}(2017{\natexlab{b}})\citenamefont {Streltsov}, \citenamefont {Rana},
  \citenamefont {Bera},\ and\ \citenamefont {Lewenstein}}]{streltsov_2017-1}%
  \BibitemOpen
  \bibfield  {author} {\bibinfo {author} {\bibfnamefont {A.}~\bibnamefont
  {Streltsov}}, \bibinfo {author} {\bibfnamefont {S.}~\bibnamefont {Rana}},
  \bibinfo {author} {\bibfnamefont {M.~N.}\ \bibnamefont {Bera}}, \ and\
  \bibinfo {author} {\bibfnamefont {M.}~\bibnamefont {Lewenstein}},\ }\href
  {\doibase 10.1103/PhysRevX.7.011024} {\bibfield  {journal} {\bibinfo
  {journal} {Phys. Rev. X}\ }\textbf {\bibinfo {volume} {7}},\ \bibinfo {pages}
  {011024} (\bibinfo {year} {2017}{\natexlab{b}})}\BibitemShut {NoStop}%
\bibitem [{\citenamefont {Zhao}\ \emph {et~al.}(2017)\citenamefont {Zhao},
  \citenamefont {Ma},\ and\ \citenamefont {Fei}}]{zhao_2017-1}%
  \BibitemOpen
  \bibfield  {author} {\bibinfo {author} {\bibfnamefont {M.-J.}\ \bibnamefont
  {Zhao}}, \bibinfo {author} {\bibfnamefont {T.}~\bibnamefont {Ma}}, \ and\
  \bibinfo {author} {\bibfnamefont {S.-M.}\ \bibnamefont {Fei}},\ }\href
  {\doibase 10.1103/PhysRevA.96.062332} {\bibfield  {journal} {\bibinfo
  {journal} {Phys. Rev. A}\ }\textbf {\bibinfo {volume} {96}},\ \bibinfo
  {pages} {062332} (\bibinfo {year} {2017})}\BibitemShut {NoStop}%
\bibitem [{\citenamefont {Wu}\ \emph {et~al.}(2017)\citenamefont {Wu},
  \citenamefont {Hou}, \citenamefont {Zhong}, \citenamefont {Yuan},
  \citenamefont {Xiang}, \citenamefont {Li},\ and\ \citenamefont
  {Guo}}]{wu_2017}%
  \BibitemOpen
  \bibfield  {author} {\bibinfo {author} {\bibfnamefont {K.-D.}\ \bibnamefont
  {Wu}}, \bibinfo {author} {\bibfnamefont {Z.}~\bibnamefont {Hou}}, \bibinfo
  {author} {\bibfnamefont {H.-S.}\ \bibnamefont {Zhong}}, \bibinfo {author}
  {\bibfnamefont {Y.}~\bibnamefont {Yuan}}, \bibinfo {author} {\bibfnamefont
  {G.-Y.}\ \bibnamefont {Xiang}}, \bibinfo {author} {\bibfnamefont {C.-F.}\
  \bibnamefont {Li}}, \ and\ \bibinfo {author} {\bibfnamefont {G.-C.}\
  \bibnamefont {Guo}},\ }\href {\doibase 10.1364/OPTICA.4.000454} {\bibfield
  {journal} {\bibinfo  {journal} {Optica}\ }\textbf {\bibinfo {volume} {4}},\
  \bibinfo {pages} {454} (\bibinfo {year} {2017})}\BibitemShut {NoStop}%
\bibitem [{\citenamefont {Wu}\ \emph {et~al.}(2018)\citenamefont {Wu},
  \citenamefont {Hou}, \citenamefont {Zhao}, \citenamefont {Xiang},
  \citenamefont {Li}, \citenamefont {Guo}, \citenamefont {Ma}, \citenamefont
  {He}, \citenamefont {Thompson},\ and\ \citenamefont {Gu}}]{wu_2017-1}%
  \BibitemOpen
  \bibfield  {author} {\bibinfo {author} {\bibfnamefont {K.-D.}\ \bibnamefont
  {Wu}}, \bibinfo {author} {\bibfnamefont {Z.}~\bibnamefont {Hou}}, \bibinfo
  {author} {\bibfnamefont {Y.-Y.}\ \bibnamefont {Zhao}}, \bibinfo {author}
  {\bibfnamefont {G.-Y.}\ \bibnamefont {Xiang}}, \bibinfo {author}
  {\bibfnamefont {C.-F.}\ \bibnamefont {Li}}, \bibinfo {author} {\bibfnamefont
  {G.-C.}\ \bibnamefont {Guo}}, \bibinfo {author} {\bibfnamefont
  {J.}~\bibnamefont {Ma}}, \bibinfo {author} {\bibfnamefont {Q.-Y.}\
  \bibnamefont {He}}, \bibinfo {author} {\bibfnamefont {J.}~\bibnamefont
  {Thompson}}, \ and\ \bibinfo {author} {\bibfnamefont {M.}~\bibnamefont
  {Gu}},\ }\href {\doibase 10.1103/PhysRevLett.121.050401} {\bibfield
  {journal} {\bibinfo  {journal} {Phys. Rev. Lett.}\ }\textbf {\bibinfo
  {volume} {121}},\ \bibinfo {pages} {050401} (\bibinfo {year}
  {2018})}\BibitemShut {NoStop}%
\bibitem [{\citenamefont {Wang}\ and\ \citenamefont
  {Renner}(2012)}]{wang_2012}%
  \BibitemOpen
  \bibfield  {author} {\bibinfo {author} {\bibfnamefont {L.}~\bibnamefont
  {Wang}}\ and\ \bibinfo {author} {\bibfnamefont {R.}~\bibnamefont {Renner}},\
  }\href {\doibase 10.1103/PhysRevLett.108.200501} {\bibfield  {journal}
  {\bibinfo  {journal} {Phys. Rev. Lett.}\ }\textbf {\bibinfo {volume} {108}},\
  \bibinfo {pages} {200501} (\bibinfo {year} {2012})}\BibitemShut {NoStop}%
\bibitem [{\citenamefont {Renes}\ and\ \citenamefont
  {Renner}(2011)}]{renes_2011}%
  \BibitemOpen
  \bibfield  {author} {\bibinfo {author} {\bibfnamefont {J.~M.}\ \bibnamefont
  {Renes}}\ and\ \bibinfo {author} {\bibfnamefont {R.}~\bibnamefont {Renner}},\
  }\href {\doibase 10.1109/TIT.2011.2162226} {\bibfield  {journal} {\bibinfo
  {journal} {IEEE Trans. Inf. Theory}\ }\textbf {\bibinfo {volume} {57}},\
  \bibinfo {pages} {7377} (\bibinfo {year} {2011})}\BibitemShut {NoStop}%
\bibitem [{\citenamefont {Tomamichel}\ and\ \citenamefont
  {Hayashi}(2013)}]{tomamichel_2013}%
  \BibitemOpen
  \bibfield  {author} {\bibinfo {author} {\bibfnamefont {M.}~\bibnamefont
  {Tomamichel}}\ and\ \bibinfo {author} {\bibfnamefont {M.}~\bibnamefont
  {Hayashi}},\ }\href {\doibase 10.1109/TIT.2013.2276628} {\bibfield  {journal}
  {\bibinfo  {journal} {IEEE Trans. Inf. Theory}\ }\textbf {\bibinfo {volume}
  {59}},\ \bibinfo {pages} {7693} (\bibinfo {year} {2013})}\BibitemShut
  {NoStop}%
\bibitem [{\citenamefont {Berta}\ \emph {et~al.}(2011)\citenamefont {Berta},
  \citenamefont {Christandl},\ and\ \citenamefont {Renner}}]{berta_2011}%
  \BibitemOpen
  \bibfield  {author} {\bibinfo {author} {\bibfnamefont {M.}~\bibnamefont
  {Berta}}, \bibinfo {author} {\bibfnamefont {M.}~\bibnamefont {Christandl}}, \
  and\ \bibinfo {author} {\bibfnamefont {R.}~\bibnamefont {Renner}},\ }\href
  {\doibase 10.1007/s00220-011-1309-7} {\bibfield  {journal} {\bibinfo
  {journal} {Commun. Math. Phys.}\ }\textbf {\bibinfo {volume} {306}},\
  \bibinfo {pages} {579} (\bibinfo {year} {2011})}\BibitemShut {NoStop}%
\bibitem [{\citenamefont {Leung}\ and\ \citenamefont
  {Matthews}(2015)}]{leung_2015}%
  \BibitemOpen
  \bibfield  {author} {\bibinfo {author} {\bibfnamefont {D.}~\bibnamefont
  {Leung}}\ and\ \bibinfo {author} {\bibfnamefont {W.}~\bibnamefont
  {Matthews}},\ }\href {\doibase 10.1109/TIT.2015.2439953} {\bibfield
  {journal} {\bibinfo  {journal} {IEEE Trans. Inf. Theory}\ }\textbf {\bibinfo
  {volume} {61}},\ \bibinfo {pages} {4486} (\bibinfo {year}
  {2015})}\BibitemShut {NoStop}%
\bibitem [{\citenamefont {Datta}\ and\ \citenamefont
  {Hsieh}(2013)}]{datta_2013}%
  \BibitemOpen
  \bibfield  {author} {\bibinfo {author} {\bibfnamefont {N.}~\bibnamefont
  {Datta}}\ and\ \bibinfo {author} {\bibfnamefont {M.~H.}\ \bibnamefont
  {Hsieh}},\ }\href {\doibase 10.1109/TIT.2012.2228737} {\bibfield  {journal}
  {\bibinfo  {journal} {IEEE Trans. Inf. Theory}\ }\textbf {\bibinfo {volume}
  {59}},\ \bibinfo {pages} {1929} (\bibinfo {year} {2013})}\BibitemShut
  {NoStop}%
\bibitem [{\citenamefont {Anshu}\ \emph {et~al.}(2017)\citenamefont {Anshu},
  \citenamefont {Devabathini},\ and\ \citenamefont {Jain}}]{anshu_2017-1}%
  \BibitemOpen
  \bibfield  {author} {\bibinfo {author} {\bibfnamefont {A.}~\bibnamefont
  {Anshu}}, \bibinfo {author} {\bibfnamefont {V.~K.}\ \bibnamefont
  {Devabathini}}, \ and\ \bibinfo {author} {\bibfnamefont {R.}~\bibnamefont
  {Jain}},\ }\href {\doibase 10.1103/PhysRevLett.119.120506} {\bibfield
  {journal} {\bibinfo  {journal} {Phys. Rev. Lett.}\ }\textbf {\bibinfo
  {volume} {119}},\ \bibinfo {pages} {120506} (\bibinfo {year}
  {2017})}\BibitemShut {NoStop}%
\bibitem [{\citenamefont {Gour}(2017)}]{gour_2017}%
  \BibitemOpen
  \bibfield  {author} {\bibinfo {author} {\bibfnamefont {G.}~\bibnamefont
  {Gour}},\ }\href {\doibase 10.1103/PhysRevA.95.062314} {\bibfield  {journal}
  {\bibinfo  {journal} {Phys. Rev. A}\ }\textbf {\bibinfo {volume} {95}},\
  \bibinfo {pages} {062314} (\bibinfo {year} {2017})}\BibitemShut {NoStop}%
\bibitem [{\citenamefont {Chitambar}\ and\ \citenamefont
  {Gour}(2016)}]{chitambar_2016}%
  \BibitemOpen
  \bibfield  {author} {\bibinfo {author} {\bibfnamefont {E.}~\bibnamefont
  {Chitambar}}\ and\ \bibinfo {author} {\bibfnamefont {G.}~\bibnamefont
  {Gour}},\ }\href {\doibase 10.1103/PhysRevLett.117.030401} {\bibfield
  {journal} {\bibinfo  {journal} {Phys. Rev. Lett.}\ }\textbf {\bibinfo
  {volume} {117}},\ \bibinfo {pages} {030401} (\bibinfo {year}
  {2016})}\BibitemShut {NoStop}%
\bibitem [{\citenamefont {Marvian}\ and\ \citenamefont
  {Spekkens}(2016)}]{marvian_2016}%
  \BibitemOpen
  \bibfield  {author} {\bibinfo {author} {\bibfnamefont {I.}~\bibnamefont
  {Marvian}}\ and\ \bibinfo {author} {\bibfnamefont {R.~W.}\ \bibnamefont
  {Spekkens}},\ }\href {\doibase 10.1103/PhysRevA.94.052324} {\bibfield
  {journal} {\bibinfo  {journal} {Phys. Rev. A}\ }\textbf {\bibinfo {volume}
  {94}},\ \bibinfo {pages} {052324} (\bibinfo {year} {2016})}\BibitemShut
  {NoStop}%
\bibitem [{\citenamefont {Yadin}\ \emph {et~al.}(2016)\citenamefont {Yadin},
  \citenamefont {Ma}, \citenamefont {Girolami}, \citenamefont {Gu},\ and\
  \citenamefont {Vedral}}]{yadin_2016}%
  \BibitemOpen
  \bibfield  {author} {\bibinfo {author} {\bibfnamefont {B.}~\bibnamefont
  {Yadin}}, \bibinfo {author} {\bibfnamefont {J.}~\bibnamefont {Ma}}, \bibinfo
  {author} {\bibfnamefont {D.}~\bibnamefont {Girolami}}, \bibinfo {author}
  {\bibfnamefont {M.}~\bibnamefont {Gu}}, \ and\ \bibinfo {author}
  {\bibfnamefont {V.}~\bibnamefont {Vedral}},\ }\href {\doibase
  10.1103/PhysRevX.6.041028} {\bibfield  {journal} {\bibinfo  {journal} {Phys.
  Rev. X}\ }\textbf {\bibinfo {volume} {6}},\ \bibinfo {pages} {041028}
  (\bibinfo {year} {2016})}\BibitemShut {NoStop}%
\bibitem [{\citenamefont {Chitambar}(2018)}]{chitambar_2018}%
  \BibitemOpen
  \bibfield  {author} {\bibinfo {author} {\bibfnamefont {E.}~\bibnamefont
  {Chitambar}},\ }\href {\doibase 10.1103/PhysRevA.97.050301} {\bibfield
  {journal} {\bibinfo  {journal} {Phys. Rev. A}\ }\textbf {\bibinfo {volume}
  {97}},\ \bibinfo {pages} {050301} (\bibinfo {year} {2018})}\BibitemShut
  {NoStop}%
\bibitem [{\citenamefont {Gour}\ and\ \citenamefont
  {Spekkens}(2006)}]{gour_2006}%
  \BibitemOpen
  \bibfield  {author} {\bibinfo {author} {\bibfnamefont {G.}~\bibnamefont
  {Gour}}\ and\ \bibinfo {author} {\bibfnamefont {R.~W.}\ \bibnamefont
  {Spekkens}},\ }\href {\doibase 10.1103/PhysRevA.73.062331} {\bibfield
  {journal} {\bibinfo  {journal} {Phys. Rev. A}\ }\textbf {\bibinfo {volume}
  {73}},\ \bibinfo {pages} {062331} (\bibinfo {year} {2006})}\BibitemShut
  {NoStop}%
\bibitem [{\citenamefont {Schr\"odinger}(1936)}]{schrodinger_1936}%
  \BibitemOpen
  \bibfield  {author} {\bibinfo {author} {\bibfnamefont {E.}~\bibnamefont
  {Schr\"odinger}},\ }\href {\doibase 10.1017/S0305004100019137} {\bibfield
  {journal} {\bibinfo  {journal} {Math.\,Pr.\,Cam.\,Phil.\,S.}\ }\textbf
  {\bibinfo {volume} {32}},\ \bibinfo {pages} {446–452} (\bibinfo {year}
  {1936})}\BibitemShut {NoStop}%
\bibitem [{\citenamefont {Hughston}\ \emph {et~al.}(1993)\citenamefont
  {Hughston}, \citenamefont {Jozsa},\ and\ \citenamefont
  {Wootters}}]{hughston_1993}%
  \BibitemOpen
  \bibfield  {author} {\bibinfo {author} {\bibfnamefont {L.~P.}\ \bibnamefont
  {Hughston}}, \bibinfo {author} {\bibfnamefont {R.}~\bibnamefont {Jozsa}}, \
  and\ \bibinfo {author} {\bibfnamefont {W.~K.}\ \bibnamefont {Wootters}},\
  }\href {\doibase 10.1016/0375-9601(93)90880-9} {\bibfield  {journal}
  {\bibinfo  {journal} {Physics Letters A}\ }\textbf {\bibinfo {volume}
  {183}},\ \bibinfo {pages} {14} (\bibinfo {year} {1993})}\BibitemShut
  {NoStop}%
\bibitem [{\citenamefont {Streltsov}\ \emph {et~al.}(2010)\citenamefont
  {Streltsov}, \citenamefont {Kampermann},\ and\ \citenamefont
  {Bru\ss{}}}]{streltsov_2010}%
  \BibitemOpen
  \bibfield  {author} {\bibinfo {author} {\bibfnamefont {A.}~\bibnamefont
  {Streltsov}}, \bibinfo {author} {\bibfnamefont {H.}~\bibnamefont
  {Kampermann}}, \ and\ \bibinfo {author} {\bibfnamefont {D.}~\bibnamefont
  {Bru\ss{}}},\ }\href {\doibase 10.1088/1367-2630/12/12/123004} {\bibfield
  {journal} {\bibinfo  {journal} {New J. Phys.}\ }\textbf {\bibinfo {volume}
  {12}},\ \bibinfo {pages} {123004} (\bibinfo {year} {2010})}\BibitemShut
  {NoStop}%
\bibitem [{\citenamefont {Uhlmann}(1976)}]{uhlmann_1976}%
  \BibitemOpen
  \bibfield  {author} {\bibinfo {author} {\bibfnamefont {A.}~\bibnamefont
  {Uhlmann}},\ }\href {\doibase 10.1016/0034-4877(76)90060-4} {\bibfield
  {journal} {\bibinfo  {journal} {Rep. Math. Phys.}\ }\textbf {\bibinfo
  {volume} {9}},\ \bibinfo {pages} {273} (\bibinfo {year} {1976})}\BibitemShut
  {NoStop}%
\bibitem [{\citenamefont {Jozsa}(1994)}]{jozsa_1994}%
  \BibitemOpen
  \bibfield  {author} {\bibinfo {author} {\bibfnamefont {R.}~\bibnamefont
  {Jozsa}},\ }\href {\doibase 10.1080/09500349414552171} {\bibfield  {journal}
  {\bibinfo  {journal} {J. Mod. Opt.}\ }\textbf {\bibinfo {volume} {41}},\
  \bibinfo {pages} {2315} (\bibinfo {year} {1994})}\BibitemShut {NoStop}%
\bibitem [{\citenamefont {Christensen}\ and\ \citenamefont
  {Vesterstr\o{}m}(1979)}]{christensen_1979}%
  \BibitemOpen
  \bibfield  {author} {\bibinfo {author} {\bibfnamefont {J.~P.~R.}\
  \bibnamefont {Christensen}}\ and\ \bibinfo {author} {\bibfnamefont
  {J.}~\bibnamefont {Vesterstr\o{}m}},\ }\href {\doibase 10.1007/BF01420337}
  {\bibfield  {journal} {\bibinfo  {journal} {Math. Ann.}\ }\textbf {\bibinfo
  {volume} {244}},\ \bibinfo {pages} {65} (\bibinfo {year} {1979})}\BibitemShut
  {NoStop}%
\bibitem [{\citenamefont {Grone}\ \emph {et~al.}(1990)\citenamefont {Grone},
  \citenamefont {Pierce},\ and\ \citenamefont {Watkins}}]{grone_1990}%
  \BibitemOpen
  \bibfield  {author} {\bibinfo {author} {\bibfnamefont {R.}~\bibnamefont
  {Grone}}, \bibinfo {author} {\bibfnamefont {S.}~\bibnamefont {Pierce}}, \
  and\ \bibinfo {author} {\bibfnamefont {W.}~\bibnamefont {Watkins}},\ }\href
  {\doibase 10.1016/0024-3795(90)90006-X} {\bibfield  {journal} {\bibinfo
  {journal} {Lin. Alg. Appl.}\ }\textbf {\bibinfo {volume} {134}},\ \bibinfo
  {pages} {63} (\bibinfo {year} {1990})}\BibitemShut {NoStop}%
\bibitem [{\citenamefont {Li}\ and\ \citenamefont {Tam}(1994)}]{li_1994}%
  \BibitemOpen
  \bibfield  {author} {\bibinfo {author} {\bibfnamefont {C.}~\bibnamefont
  {Li}}\ and\ \bibinfo {author} {\bibfnamefont {B.}~\bibnamefont {Tam}},\
  }\href {\doibase 10.1137/S0895479892240683} {\bibfield  {journal} {\bibinfo
  {journal} {SIAM J. Matrix Anal. \& Appl.}\ }\textbf {\bibinfo {volume}
  {15}},\ \bibinfo {pages} {903} (\bibinfo {year} {1994})}\BibitemShut
  {NoStop}%
\bibitem [{\citenamefont {Watrous}(2009)}]{watrous_2009}%
  \BibitemOpen
  \bibfield  {author} {\bibinfo {author} {\bibfnamefont {J.}~\bibnamefont
  {Watrous}},\ }\href {\doibase 10.4086/toc.2009.v005a011} {\bibfield
  {journal} {\bibinfo  {journal} {Theory Comput.}\ }\textbf {\bibinfo {volume}
  {5}},\ \bibinfo {pages} {217} (\bibinfo {year} {2009})}\BibitemShut {NoStop}%
\bibitem [{\citenamefont {Watrous}(2013)}]{watrous_2013}%
  \BibitemOpen
  \bibfield  {author} {\bibinfo {author} {\bibfnamefont {J.}~\bibnamefont
  {Watrous}},\ }\href {\doibase 10.4086/cjtcs.2013.008} {\bibfield  {journal}
  {\bibinfo  {journal} {Chic.\,J.\,Th.\,Comp.\,Sci.}\ }\textbf {\bibinfo
  {volume} {19}},\ \bibinfo {pages} {1} (\bibinfo {year} {2013})}\BibitemShut
  {NoStop}%
\bibitem [{\citenamefont {Uhlmann}(2010)}]{uhlmann_2010}%
  \BibitemOpen
  \bibfield  {author} {\bibinfo {author} {\bibfnamefont {A.}~\bibnamefont
  {Uhlmann}},\ }\href {\doibase 10.3390/e12071799} {\bibfield  {journal}
  {\bibinfo  {journal} {Entropy}\ }\textbf {\bibinfo {volume} {12}},\ \bibinfo
  {pages} {1799} (\bibinfo {year} {2010})}\BibitemShut {NoStop}%
\bibitem [{\citenamefont {Vijayan}\ \emph {et~al.}(2018)\citenamefont
  {Vijayan}, \citenamefont {Chitambar},\ and\ \citenamefont
  {Hsieh}}]{vijayan_2018}%
  \BibitemOpen
  \bibfield  {author} {\bibinfo {author} {\bibfnamefont {M.~K.}\ \bibnamefont
  {Vijayan}}, \bibinfo {author} {\bibfnamefont {E.}~\bibnamefont {Chitambar}},
  \ and\ \bibinfo {author} {\bibfnamefont {M.-H.}\ \bibnamefont {Hsieh}},\
  }\href {\doibase 10.1088/1751-8121/aadc21} {\bibfield  {journal} {\bibinfo
  {journal} {J. Phys. A: Math. Theor.}\ }\textbf {\bibinfo {volume} {51}},\
  \bibinfo {pages} {414001} (\bibinfo {year} {2018})}\BibitemShut {NoStop}%
\end{thebibliography}%
